\documentclass[format=acmsmall, review=false, screen=true]{acmart}

\usepackage{booktabs} 

\acmJournal{TKDD} 
\acmVolume{3}
\acmNumber{1}
\acmArticle{18}
\acmYear{2019}
\acmMonth{8}

\setcopyright{acmcopyright}

\acmDOI{0000001.0000001}

\usepackage{multirow}
\usepackage[noend]{algpseudocode}
\usepackage[ruled]{algorithm2e} 

\SetAlFnt{\small}
\SetAlCapFnt{\small}
\SetAlCapNameFnt{\small}
\SetAlCapHSkip{0pt}
\IncMargin{-\parindent}

\usepackage{array}
\newcolumntype{P}[1]{>{\centering\arraybackslash}p{#1}}
\newcolumntype{M}[1]{>{\centering\arraybackslash}m{#1}}
\newcolumntype{R}[1]{>{\arraybackslash}m{#1}}

\definecolor{orange}{rgb}{1,0.5,0}
\definecolor{graynode}{RGB}{20,20,20}
\definecolor{crimsonred}{RGB}{220,20,60}
\definecolor{darkgraynode}{gray}{0.5}
\definecolor{lightgraynode}{gray}{0.8}

\usepackage{rotate}
\usepackage{capt-of}
\usepackage{tabulary}
\usepackage{setspace}
\usepackage{amssymb}
\usepackage{pifont}

\definecolor{gray}{RGB}{20,20,20}
\definecolor{gray}{RGB}{0.7,0.7,0.7}
\definecolor{greencm}{RGB}{0,153,0}

\definecolor{plotblue}{RGB}	{30,144,255}
\definecolor{plotgreen}{RGB}	{50,205,50}
\definecolor{plotred}{RGB}	{220,20,60}

\definecolor{myyellow}{RGB}{255,255,204}
\definecolor{myred}{RGB}{255,204,204}
\definecolor{myblue}{RGB}{0,200,255}
\definecolor{mygreen}{RGB}{80,220,80}

\newcommand*\hrulefillvar[1][0.4pt]{\leavevmode\leaders\hrule height#1\hfill\kern0pt}

\usepackage{ragged2e}

\usepackage{comment}
\usepackage{tabularx}

\usepackage{subfigure}
\usepackage{enumitem}

\definecolor{thedarkblue}{RGB}{0,0,120} 
\definecolor{mydarkblue}{rgb}{0,0.08,0.45} 

\usepackage{hyperref}
\hypersetup{%
    colorlinks=true,
    linkcolor=mydarkblue,
    citecolor=mydarkblue,
    filecolor=mydarkblue,
    urlcolor=mydarkblue}
    
\usepackage{microtype}
\usepackage{graphicx}

\usepackage{balance}
\usepackage{tabularx}

\usepackage{booktabs}
\usepackage{tabularx}
\usepackage{comment}

\begin{document}

\title{Utility Mining Across Multi-Sequences with Individualized Thresholds}

\author{Wensheng Gan}
\affiliation{%
	\institution{Harbin Institute of Technology (Shenzhen)}
	\city{Shenzhen}
	\country{China}
}
\email{wsgan001@gmail.com}

\author{Jerry Chun-Wei Lin}
\authornote{This is the corresponding author}
\affiliation{%
	\institution{Western Norway University of Applied Sciences (HVL)}
	\city{Bergen}
	\country{Norway}
}
\email{jerrylin@ieee.org}

\author{Jiexiong Zhang}
\affiliation{%
	\institution{Harbin Institute of Technology (Shenzhen)}
	\city{Shenzhen}
	\country{China}	
}
\email{jiexiong.zhang@foxmail.com}

\author{Philip S. Yu}
\affiliation{%
	\institution{University of Illinois at Chicago}
	\city{Chicago}
	\country{USA}
}
\email{psyu@uic.edu}

\renewcommand\shortauthors{W. Gan et al.}

\begin{abstract}

Utility-oriented pattern mining has become an emerging topic since it can reveal high-utility patterns (e.g., itemsets, rules, sequences) from different types of data, which provides more information than the traditional frequent/confident-based pattern mining models. The utilities of various items are not exactly equal in realistic situations; each item has its own utility or importance. In general, user considers a uniform minimum utility (\textit{minutil}) threshold to identify the set of high-utility sequential patterns (HUSPs). This is unable to find the interesting patterns while the \textit{minutil} is set extremely high or low. We first design a new utility mining framework namely USPT for mining  high-\textbf{\underline{U}}tility \textbf{\underline{S}}equential \textbf{\underline{P}}atterns across multi-sequences with individualized \textbf{\underline{T}}hresholds. Each item in the designed framework has its own specified minimum utility threshold. Based on the lexicographic-sequential tree and the utility-array structure, the USPT framework is presented to efficiently discover the HUSPs. With the upper-bounds on utility, several pruning strategies are developed to  prune the unpromising candidates early in the search space. Several experiments are conducted on both real-life and synthetic datasets to show the performance of the designed USPT algorithm, and the results showed that USPT could achieve good effectiveness and efficiency for mining HUSPs with individualized minimum utility thresholds. 

\end{abstract}

%
%




\maketitle

\section{Introduction}

Sequential pattern mining (SPM) \cite{agrawal1995mining,srikant1996mining,pei2001prefixspan} has become an interesting and emerging issue in knowledge discovery in databases (KDD) \cite{agrawal1993database,chen1996data}, and it has been used in many real-life applications, such as behavioral analysis of customers, DNA sequence analysis, and natural disaster analysis \cite{fournier2017survey}. The purpose of SPM is to discover the frequent sequences from the sequence database using a uniform users' defined minimum threshold. Both SPM and frequent itemset mining (FIM) \cite{agrawal1994fast,han2004mining} are frequent pattern mining approaches, where the main difference between them is that the processed data in SPM is consequentially time-ordered. Thus, it is more complicated and complex to discover the interesting information from the sequential databases in SPM. 

The frequency-based pattern mining frameworks (e.g., FIM and SPM) mainly consider the co-occurrence frequencies of the itemset/sequence. The above well-known frameworks cannot, however, retrieve the high-utility patterns or handle more important factors for pattern mining, such as interestingness, utility, importance, and risk. To solve this limitation, in recent decades, the utility-oriented pattern mining framework \cite{2gan2018survey} named high-utility itemset mining (HUIM) has been extensively studied and successfully applied in numerous domains \cite{chan2003mining,tseng2013efficient,liu2005two}. HUIM considers the quantity and unit utility of items to measure the high-utility itemsets (HUIs) from a quantitative dataset. Traditional HUIM does not hold the downward closure property \cite{agrawal1994fast}, and it becomes a non-trivial task to mine the useful and meaningful HUIs from the databases.

Although HUIM is useful to discover high-utility patterns in many real-life applications, it cannot handle a sequence database that contains the embedded timestamp information of the events/items. In particular, the ordering of elements or events is important in many real-life domains. For instance, HUIM can be applied to analyze the shopping behaviors of the customers for targeted marketing and also be utilized for the recommendation task since the order of purchase products should be considered to retrieve more interesting patterns. Besides, the events/sequences in sequential data commonly contain multiple time-order sequences, but not a single long sequence. For example, the real dataset in market basket analysis is usually has rich information, including user ID, timestamp, event ID, record (event), quantity, and unit price. In general, each record in these datasets contains multi-sequences. Episode mining \cite{mannila1997discovery} focuses on discovering interesting patterns from a single sequence, while SPM \cite{agrawal1995mining,srikant1996mining,pei2001prefixspan,fournier2017survey} usually discovers frequent sequences from multi-sequences.  To solve the above limitation, high-utility sequential pattern mining (HUSPM) \cite{lan2014applying,wang2016efficiently,yin2012uspan,yin2013efficiently} was introduced, and it mainly aims at discovering more informative sequential patterns. It is more complicated than HUIM and SPM since both the ordering and the utility of events or elements are considered together. A sequence is considered to be a high-utility sequential pattern (HUSP) if its overall utility is no less than a user-specified minimum utility threshold. To retrieve HUSPs, several techniques and approaches have been developed in the past \cite{alkan2015crom,wang2016efficiently,yin2012uspan}. However, most of these techniques focus on improving the mining efficiency, but not its effectiveness. In addition, all the existing HUSPM approaches suffer from an important limitation since most of them focus on mining the HUSPs with a uniform and single minimum utility threshold. Thus, all the items and sequences are considered as being the same weight or importance, which is unfair and not realistic in some real-life situations.  

In real-life situations, each object/item has its own importance and utility because the utility (e.g., risk, profit) for various objects/items are not exactly equal. Some patterns with higher importance may have low utilities in databases. For example, a shopping mall contains more than hundred or thousand products, and it is impossible to measure them with the same importance due to their characteristics. For instance, diamond and clothe are two different products and should not be treated as the same importance. It is more reasonable to respectively set different utility thresholds for diamond and clothe since they have different features. Although there are numerous studies of market basket analysis focusing on SPM from the sequential purchase data, none of them consider the quantity, utility factor, or individuality of importance at the same time.  Another example can be referred to the recommendation system, which utilizes time-stamp information to improve the performance. A lower threshold could be set for a target topic or product that users are interested to reveal more information, or to sell multiple products for mixed bundling.

Therefore, mining HUSPs with an individualized threshold of each item is more realistic than existing HUSPM algorithms for retrieving more useful and meaningful patterns. However, it is not a trivial task to address these challenges. Therefore, we are motivated to first develop a novel utility-oriented mining framework, called mining high-\textbf{\underline{U}}tility \textbf{\underline{S}}equential \textbf{\underline{P}}atterns with individualized \textbf{\underline{T}}hresholds (abbreviated as USPT). The major contributions of this paper are described below:

\begin{enumerate}
	\item To the best of our knowledge, this is first work to formulate a new USPT framework for mining high-utility sequential patterns across multi-sequences with individualized thresholds.  The proposed USPT method sets an individualized threshold for each item, instead of a uniform minimum utility threshold for all items. This progress can avoid the ``rare item'' problem \cite{liu1999mining}, and more flexible for real-life situations.
	\item A compressed structure, called utility-array, is utilized to keep the necessary information in a lexicographic-sequential tree. Based on the properties of utility in the projection-based utility-array, several pruning strategies are further designed to improve the mining performance.
	\item  Experiments indicated that the USPT algorithm and its variations achieved effectiveness and efficiency for mining the set of HUSPs with individualized thresholds from sequential databases. More specifically, the proposed  strategies can effectively prune the unpromising candidates and speed up mining process.
\end{enumerate}

Note that some key concepts were presented in a preliminary version \cite{lin2017high} of this paper. The rest of this paper is organized as follows. A literature review is given in Section \ref{sec:2}. Preliminaries and the problem statement of the novel USPT are stated in Section \ref{sec:3}. The developed USPT method with utility-array and pruning strategies are discussed in Section \ref{sec:4}. An illustrated example, which can illustrate the proposed algorithm step-by-step, is provided in Section \ref{sec:5}. Several experiments of the designed algorithm are described in Section \ref{sec:6}. Conclusions and directions for future works are given in Section \ref{sec:7}.

\section{Related Work}
\label{sec:2}
We structure the literature review into three main elements, consisting of high-utility itemset mining, high-utility sequential pattern mining, and  pattern mining across multi-sequences or with multiple/individualized thresholds.

\subsection{High-Utility Itemset Mining}

In 2003, Chan \textit{et al.} \cite{chan2003mining} considered the utility theory \cite{marshall2005principles} in a pattern mining task, and introduced the problem of high-utility itemset mining (HUIM). Yao \textit{et al.} \cite{yao2004foundational} then introduced an approach for mining the set of high-utility itemsets (HUIs). Since traditional HUIM does not hold the downward closure property \cite{agrawal1994fast}, Liu \textit{et al.} \cite{liu2005two} developed a transaction-weighted downward closure property and presented the concept of high transaction-weighted utilization itemsets (HTWUIs). To obtain better performance for mining HUIs, several tree-based algorithms have been designed, such as IHUP \cite{ahmed2009efficient}, HUP-growth \cite{lin2011effective}, UP-Growth \cite{tseng2010up}, and UP-Growth$^{+}$ \cite{tseng2013efficient}. To reduce a huge number of candidates, Liu \textit{et al.} \cite{liu2012mining} developed the HUI-Miner algorithm to efficiently discover the set of HUIs by using a novel utility-list structure. This process can directly discover HUIs. Thus, the procedures of candidate generation and multiple database scans can be avoided. Variations of HUIM have been extensively designed and studied to more efficiently discover HUIs, such as FHM \cite{fournier2014fhm} and EFIM \cite{zida2015efim}.

In addition, a number of HUIM algorithms have been extensively studied to improve the mining effectiveness. For example, Tseng \textit{et al.} studied the problem of concise and lossless representation of high-utility itemsets \cite{tseng2015efficient}, as well as the top-\textit{k} issue of HUIM \cite{tseng2016efficient}. Lin \textit{et al.} introduced several models to discover various types of high-utility patterns from different types of data, such as mining HUIs from uncertain databases \cite{lin2016efficient}, mining up-to-date HUIs \cite{lin2015efficient}, and dynamic utility mining \cite{lin2015fast,1gan2018survey}. Other interesting issues, such as mining on-shelf HUIs \cite{lan2011discovery}, mining high-utility association rules \cite{mai2017lattice}, and mining high-utility occupancy patterns \cite{gan2019huopm}, are also extensively explored. Ryang at al. then addressed the issue that mining HUIs from data stream \cite{ryang2016high}, and Yun \textit{et al.} extended the HUIM problem for establishing manufacturing plan \cite{yun2017efficient}.

\subsection{High-Utility Sequential Pattern Mining}

Sequential pattern mining (SPM) \cite{agrawal1995mining,pei2001prefixspan,srikant1996mining} provides more information to handle order-based applications, such as behavior analysis, DNA sequence analysis, and weblog mining. SPM was first proposed in \cite{agrawal1995mining} and has been extensively studied, resulting in methods such as GSP \cite{srikant1996mining}, PrefixSpan \cite{pei2001prefixspan}, SPADE \cite{zaki2001spade}, and SPAM \cite{ayres2002sequential}. However, SPM is based on a frequency/support framework for mining frequent sequences, which does not focus on business interests. Thus, a novel utility-oriented mining framework, called high-utility sequential pattern mining (HUSPM) \cite{alkan2015crom,wang2016efficiently,yin2012uspan}, has been developed. HUSPM considers ordered sequences and reveals the utilities of multi-sequences. It has many real applications, including high-utility patterns of purchase behavior \cite{gan2019proum}, high-utility web access sequences \cite{ahmed2010mining}, and high-utility mobile sequential patterns \cite{shie2011mining}. Ahmed \textit{et al.} \cite{ahmed2010novel} first developed level-wise (UL) and pattern-growth (US) approaches for HUSPM. However, both of them can only handle simple sequences. Next, an efficient USpan algorithm \cite{yin2012uspan} was proposed, and it employs a lexicographic quantitative sequence (LQS)-tree to ensure that the complete set of HUSPs can be discovered. The utility information of each node in the LQS-tree is stored in the utility-matrix for mining HUSPs without multiple database scans. To obtain an easier parameter setting, a TUS algorithm \cite{yin2013efficiently} was proposed to discover the top-\textit{k} HUSPs. Lan \textit{et al.} \cite{lan2014applying} proposed a HUSPM approach using the projection mechanism. It uses the sequence-utility upper-bound (\textit{SUUB}) to over-estimate the upper-bound value of the sequence, which is based on the maximum utility measure. The introduced maximum utility measure can obtain a reduced set of patterns, which contains the complete but compact information from the original set. To improve the mining performance for HUSPM, Alkan \textit{et al.} \cite{alkan2015crom} proposed the HuspExt algorithm by estimating the Cumulated Rest of Match (\textit{CRoM}) as the upper-bound value to promptly prune unpromising candidates. Wang \textit{et al.} \cite{wang2016efficiently} developed two tight utility upper-bounds called prefix extension utility (\textit{PEU}) and reduced sequence utility (\textit{RSU}), which can be used to speed up the mining process. Recently, two novel dynamic models, named  IncUSP-Miner \cite{wang2018incremental} and IncUSP-Miner+ \cite{wang2018incremental}, were proposed to incrementally discover HUSPs. A projection-based ProUM \cite{gan2019proum} algorithm introduces a new data structure namely utility-array, and has the state-of-the-art performance for mining high-utility sequential patterns. Gan \textit{et al.} \cite{2gan2018survey} presented a comprehensive report of utility-oriented pattern mining and the up-to-date developments.

\subsection{Pattern Mining across Multi-Sequences or with Multiple Thresholds}

In the past, several works focused on pattern mining across multi-sequences \cite{pei2001prefixspan,zaki2001spade,lee2009mining,fournier2017survey}. Different from the early studies that discover interesting patterns from a single sequence (most of them are related to episode mining \cite{mannila1997discovery}), the task of  SPM aims at mining frequent sequential patterns from multi-sequences where each record contains multiple sequences \cite{pei2001prefixspan,zaki2001spade}. In general, each record in a sequential database, also called an event/sequence, contains multiple time-series sequences, but not a single long sequence. Up to now, many algorithms for SPM have been extensively studied \cite{fournier2017survey,3gan2018survey}. In addition, Pinto \textit{et al.} \cite{pinto2001multi} first proposed the theme of multi-dimensional sequential pattern mining, which is different from the general sequential pattern mining. The AprioriMD \cite{yu2005mining} and PrefixMDSpan \cite{yu2005mining} algorithms can find sequential patterns from $d$-dimensional sequence data, where $d$ > 2. Pinto's work about multi-dimensional SPM aims to find frequent subsequences in subsets of the data, while AprioriMD finds subsequences on different granularities of sequences. Lee \textit{et al.} \cite{lee2009mining} developed the CMP-Miner algorithm to discover closed patterns from the multiple time-series database.

The above algorithms only consider a uniform minimum support or utility threshold to mine the required information. However, objects/items do not have the same nature (e.g., frequency, utility, weight, and risk) in realistic applications. It is an unfair measurement to consider them as having the same importance to discover the meaningful information. Liu \textit{et al.} first proposed multiple minimum supports for association rule mining (ARM) to solve the ``\textit{rare item}'' problem \cite{liu1999mining}. The proposed MSApriori algorithm adopts the sorted closure property, which is used to prune the large candidates \cite{liu1999mining}. Liu \textit{et al.} \cite{liu2011discovering} extended the multiple thresholds to the SPM method and proposed the multiple supports-generalized sequential pattern (MS-GSP) algorithm to find the frequent sequential patterns with multiple minimum supports. This approach is performed in a level-wise manner, and thus a large amount of computation is required. Hu \textit{et al.} \cite{hu2013efficient} then considered the multiple minimum support to present a more efficient MSCP-growth for SPM.

Recently, the problem of high-utility itemset mining with individualized utility thresholds has been studied, including the Apriori-like HUI-MMU algorithm \cite{2lin2016efficient}, the one-phase HIMU algorithm \cite{gan2016more}, and MHUI algorithm \cite{krishnamoorthy2018efficient}. HUI-MMU \cite{2lin2016efficient} first introduced the \textit{MIU} (minimum utility threshold of a $k$-itemset) concept and formulates the problem statement of utility mining with multiple utility thresholds, while its mining performance is worse. Both HIMU \cite{gan2016more} and the enhanced MHUI \cite{krishnamoorthy2018efficient} utilize the vertical data structure utility-list \cite{liu2012mining} to store the necessary information of utility, and adopt several powerful pruning strategies to reduce the search space and memory cost. Although the above algorithms involve multiple thresholds (e.g., support or utility)  to pattern mining, there is no work related to  HUSPM with individualized thresholds yet.

\section{Preliminary and Problem Statement}
\label{sec:3}

In this section, some concepts and principles of utility mining on sequence with individualized thresholds are firstly presented.  We then formulate the problem statement of high-utility sequential pattern mining with individualized thresholds as follows.

Let \textit{I} = \{$ i_{1} $, $ i_{2} $, $\dots$, $ i_{m} $\} be a finite set of distinct items. An \textit{itemset} is denoted as $ w $ = [$ i_{1} $, $ i_{2} $, $\dots$, $ i_{c} $], which is a subset of $ I $ without quantities. A \textit{quantitative itemset}, denoted as \textit{v} = [($ i_{1}, q_{1} $) ($ i_{2}, q_{2} $) $\dots$ ($ i_{c}, q_{c} $)], is a subset of \textit{I}. In addition, each item $ i_{c}$ within the quantitative itemset has its own quantity $q_{c} $, called \textit{internal utility}.  Without a loss of generality, we assume that items in an itemset (quantitative itemset) are listed in \textit{alphabetical} order since items are unordered in an itemset (quantitative itemset). A \textit{sequence} is an ordered sequence with at least one or more itemsets but without quantities, which can be denoted as:  $ t $ = $ < $$ w_{1} $, $ w_{2} $, $\dots$, $ w_{d} $$ > $.  A \textit{quantitative sequence} is an ordered sequences with at least one or more quantitative itemsets, which can be denoted as: $ s $ = $ < $$ v_{1} $, $ v_{2} $, $\dots$, $ v_{d} $$ > $. In this paper, we use ``\textit{q}-'' instead of \textit{quantitative}. Thus, a ``\textit{q}-sequence'' represents the set of ordered sequences, and each item within a sequence has its own quantity. The ``sequence'' directly represents the set of ordered sequences without quantities. We also can extend this concept to the ``\textit{q}-itemset'' and ``itemset''. Thus, ``\textit{q}-itemset'' represents a set of items, and each item has its own quantity; and ``itemset'' only represents the set of items without quantities. For instance, $ < $[(\textit{b}, 2) (\textit{c}, 1)], [(\textit{a}, 3)]$ > $ is a \textit{q}-sequence, while $ < $[\textit{bc}], [\textit{a}]$ > $ is a sequence. [(\textit{b}, 2) (\textit{c}, 1)] is a \textit{q}-itemset and [\textit{bc}] is an itemset.

A \textit{sequential database with quantitative values} is a set of sequences and is denoted as: \textit{QSD} = \{$ S_{1} $, $ S_{2} $, $\dots$, $ S_{n} $\}. Each sequence  $ S_{q}\in QSD$ is a \textit{q}-sequence and has as a unique identifier, which is used to present its \textit{SID}. Moreover, each item in \textit{QSD} is associated with an \textit{external utility}, which can be the unit utility of each item and denoted as $ pr(i_{j}) $. In this paper, Table \ref{table:db} is used as an example to illustrate the definitions and processed steps. The \textit{utility-table} is used to present the unit utility of each item shown in Table \ref{table:profit}. Based on the illustrated example, [(\textit{a}:3) (\textit{b}:2)] is the first \textit{q}-itemset in $ S_{3} $. The quantity of (\textit{b}) in this \textit{q}-itemset is set as 2, and its utility is measured as 2 $\times$ \$5 = \$10.

\begin{table}[!htbp]
	\setlength{\abovecaptionskip}{0pt}
	\setlength{\belowcaptionskip}{0pt} 
	\caption{A quantitative sequential database (QSD).}
	\centering
	\begin{tabular}{|c|c|}
		\hline
		\textbf{SID}	 & \textbf{Q-sequence} \\ \hline
		$ S_{1} $ & $ < $[(\textit{d}:2)], [(\textit{a}:2) (\textit{e}:3)], [(\textit{b}:3) (\textit{d}:3)], [(\textit{c}:4) (\textit{e}:5)]$ > $  \\ \hline
		
		$ S_{2} $ & $ < $[(\textit{b}:1) (\textit{d}:3)], [(\textit{b}:5) (\textit{c}:3) (\textit{e}:2)], [(\textit{a}:1) (\textit{c}:2) (\textit{d}:3) (\textit{e}:4)]$ > $  \\ \hline

		$ S_{3} $ & $ < $[(\textit{a}:3) (\textit{b}:2)], [(\textit{a}:2) (\textit{b}:3) (\textit{c}:1)], [(\textit{b}:4) (\textit{c}:5) (\textit{e}:4)], [(\textit{d}:3)]$ > $  \\ \hline
		
		$ S_{4} $ & $ < $[(\textit{a}:3) (\textit{c}:2)], [(\textit{b}:5) (\textit{d}:1)], [(\textit{c}:4) (\textit{d}:5) (\textit{e}:3)]$ > $  \\ \hline
		
		$ S_{5} $ & $ < $[(\textit{a}:3) (\textit{b}:4)], [(\textit{a}:2) (\textit{b}:2) (\textit{c}:5) (\textit{d}:3) (\textit{e}:4)]$ > $  \\ \hline
		
		$ S_{6} $ & $ < $[(\textit{f}:4)], [(\textit{b}:5) (\textit{c}:3)], [(\textit{b}:3) (\textit{e}:4)]$ > $  \\ \hline
	\end{tabular}
	\label{table:db}
\end{table}

\begin{table}[!htbp] 
	\setlength{\abovecaptionskip}{0pt}
	\setlength{\belowcaptionskip}{0pt} 		
	\caption{A utility table.}
	\centering
	\begin{tabular}{|c|c|c|c|c|c|c|}
		\hline
		\textbf{Item} &	\textit{a} & \textit{b} & \textit{c} &	\textit{d} & \textit{e} & \textit{f} \\ \hline
		\textbf{Utility} & \$4 & \$5 & \$3 & \$1 & \$2 & \$6 \\ \hline
	\end{tabular}
	\label{table:profit}	
\end{table}

\begin{definition}C
	\rm Let $ mu(i_{j}) $ represent the minimum utility threshold of $ i_{j} $ in a \textit{QSD}, where $ i_{j} $ is an item. A multiple threshold table is denoted as \textit{M-table}, which consists of the individualized utility threshold of each item and defined as: $ \textit{M-table}$ = $\{mu(i_{1}), mu(i_{2}), \dots, mu(i_{m})\}$. In general, the \textit{M-table} is user-specified based on prior knowledge.
\end{definition}

In the running example, we assume that the \textit{M-table} is defined as: \textit{M-table} = \{$ mu(a) $, $ mu(b) $, $ mu(c) $, $ mu(d) $, $ mu(e) $, $ mu(f) $\} = \{\$500, \$500, \$500, \$500, \$200, \$70\}.

\begin{definition}
	\rm The minimal utility threshold of a sequence $ t $ with respect to items among $t$  is denoted as $ MIU(t) $, which is used to represent the minimal  utility (\textit{mu}) value of all items in \textit{t}. Thus, it is defined as: $MIU(t) = min\{mu(i_{j})|i_{j}\in t\}$. Note that \textit{t} represents a sequence throughout this paper.
	
\end{definition}

Obviously, the \textit{MIU} value of a sequence is set based on the user-specified \textit{M-table}. For example, $ MIU $($<$[\textit{b}]$>$) = \textit{min}\{$ mu(b) $\} = \$500, and $ MIU $($<$[\textit{be}]$>$) = \textit{min}\{$ mu(b$), $ mu(e) $\} = \textit{min}\{\$500, \$200\} = \$200.

\begin{definition}
	\rm Let $ u(i_{j}, v) $ denote the utility of $ i_{j} $ in a \textit{v}, where $ i_{j} $ is an item and \textit{v} is an $q$-itemset. It is defined as: $u(i_{j}, v)$ = $q(i_{j}, v)\times pr(i_{j})$. The $ q(i_{j}, v) $ represents the quantity of ($ i_{j} $) in $ v $, and the unit utility of ($ i_{j} $) represents as $ pr(i_{j}) $. $ u(v) $ is the utility of a \textit{q}-itemset $ v $, which can be defined as $u(v) = \sum_{i_{j}\in v}u(i_{j}, v)$.
\end{definition}

For example, in Table \ref{table:db}, the utility of an item (\textit{a}) in the first \textit{q}-itemset of $ S_{3} $ is calculated as: $u(a$, [(\textit{a}:3) (\textit{b}:2)]) = $ q(a$, [(\textit{a}:3) (\textit{b}:2)]) $\times $ $pr(a)$ = 3 $\times$ \$4 = \$12. And \textit{u}([(\textit{a}:3) (\textit{b}:2)]) = \textit{u}(\textit{a}, [(\textit{a}:3) (\textit{b}:2)]) + \textit{u}(\textit{b}, [(\textit{a}:3) (\textit{b}:2)]) = 3 $\times$ \$4 + 2 $\times$ \$5 = \$22.

\begin{definition}	
	\rm Let $ u(s) $ denote the overall utility of a \textit{q}-sequence such that $ s $ = $<$$v_{1}, v_{2}, \dots, v_{d}$$>$, and it is defined as $u(s)$ = $\sum_{v\in s} u(v)$. Let $u(QSD)$ denote the overall utility of a  quantitative sequential database \textit{QSD}, and it can be defined as: $u(QSD)$ = $\sum_{s\in QSD} u(s)$.	
\end{definition}

For example, $ u(S_{3}) $ = \textit{u}([(\textit{a}:3) (\textit{b}:2)]) + \textit{u}([(\textit{a}:2)] (\textit{b}:3) (\textit{c}:1)) + \textit{u}([(\textit{b}:4) (\textit{c}:5) (\textit{e}:4)]) + \textit{u}([(\textit{d}:3)]) = \$22 + \$26 + \$43 + \$3 = \$94. Thus, the $ u(QSD) $ is to sum up the utilities of $ S_{1} $ to $ S_{6} $ as: \$56 + \$67 + \$94 + \$67 + \$76 + \$81 = \$441.

\begin{definition}
	\rm Let $ v $ and $ v' $ represent two \textit{q}-itemsets, and $ v\subseteq v' $ denote $ v $ is contained in $ v' $, if for every item in $ v $, it has the same item having the same quantity in $ v' $. An itemset (\textit{q}-itemset) containing \textit{k} items is named \textit{k}-itemset (\textit{k}-\textit{q}-itemset). A sequence (\textit{q}-sequence) containing \textit{k} items is named \textit{k}-sequence (\textit{k}-\textit{q}-sequence). 
\end{definition}

For example, in Table \ref{table:db}, an itemset [\textit{ab}] is contained in the itemset [\textit{abc}]. The \textit{q}-itemset [(\textit{a}:3) (\textit{b}:2)] is contained in [(\textit{a}:3) (\textit{b}:2) (\textit{c}:1)] and [(\textit{a}:3) (\textit{b}:2) (\textit{d}:2)], but it is not contained in [((\textit{a}:1) \textit{b}:2) (\textit{c}:3)] and [(\textit{a}:3) (\textit{b}:4) (\textit{e}:4)]. And the sequences $<$[(\textit{b}:2)], [(\textit{d}:3)]$>$ and $<$[(\textit{b}:4)], [(\textit{d}:3)]$>$ are contained in $ S_{3} $, but $<$[(\textit{b}:1)], [\textit{d}:3]$>$ and $<$[(\textit{b}:4)], [(\textit{d}:4)]$>$ are not contained in $ S_{3} $. In addition, $ S_{3} $ is a 9-\textit{q}-sequence and the first \textit{q}-itemset in $ S_{3} $ is a 2-\textit{q}-itemset.

\begin{definition}
	\rm Given a \textit{q}-sequence \textit{s} = $<$$v_{1}, v_{2}, \dots, v_{d}$$>$ and a sequence $ t $ = $<$$w_{1}, w_{2}, \dots, w_{d'}$$>$, if $ d $ = $ d' $ and the items in $ v_{k} $ are equal to the items in $ w_{k} $ such that $ 1\leq k\leq d $, and $ t $ matches $ s $, then it is denoted as: $ t \sim s $.	
\end{definition}

For example, in Table \ref{table:db}, $<$[\textit{ab}], [\textit{abc}], [\textit{bce}]$>$ matches $ S_{3} $. However, we may have multiple matches of a target sequence in a \textit{q}-sequence. For instance, $<$[\textit{b}], [\textit{c}]$>$ has three matches in $ S_{3} $: $<$[\textit{b}:2], [\textit{c}:1]$>$, $<$[\textit{b}:2], [\textit{c}:5]$>$ and $<$[\textit{b}:3], [\textit{c}:5]$>$. Thus, efficiently discovering the required information in the developed USPT framework is a non-trivial task.

\begin{definition}
	\rm Given two sequences, \textit{t} = $<$$w_{1}, w_{2}, \dots, w_{d}$$>$ and $ t' $ = $<$$w'_{1}, w'_{2}, \dots, w'_{d'}$$>$, $ t $ is said to be contained in $ t' $ as $ t \subseteq t $' if it obtains a sequence $ 1\leq k_{1}\leq k_{2}\leq\dots\leq d' $ such that $ w_{j}\subseteq w'_{k_{j}} $ for $ 1\leq j\leq d $. 
	Given two \textit{q}-sequences $ s $ = $<$$v_{1}, v_{2}$, $\dots, v_{d}$$>$ and $ s' $ = $<$$v'_{1}, v'_{2}$, $\dots, v'_{d'}$$>$, let $ s\subseteq s' $ represent that $ s $ is contained in $s'$ if it has an integer sequence $ 1\leq k_{1}\leq k_{2} \leq\dots\leq d' $ such that $ v_{j}\subseteq v'_{k_{j}} $ for $ 1\leq j \leq d $. Thus, $ t \subseteq s $ can be denoted as  $ t \sim s_{k} \wedge s_{k} \subseteq s $ for convenience.
\end{definition}

According to the previous studies, the maximal utility \cite{wang2016efficiently,yin2012uspan} is adopted to measure the utility of a sequence in a sequential database. Note that our proposed USPT model follows this definition of utility, which is defined as follows.

\begin{definition}
	\rm Let $ u(t, s) $ denote the utility of a sequence $ t $ in a \textit{q}-sequence \textit{s}, and it can be defined as: $u(t, s)$ = $max\{u(s_{k})|t \sim s_{k} \wedge s_{k} \subseteq s\}$. The overall utility of \textit{t} sequence in \textit{QSD} is denoted as  $ u(t) $, and defined as: $u(t)$ = $\sum_{s\in QSD} \{u(t,s)|t \subseteq s\}$.
\end{definition}

For example, $ u($$<$$[b], [c]$$>$$, S_{3}) $ = $ max$\{\textit{u}($<$[\textit{b}:2], [\textit{c}:1]$>$), \textit{u}($<$[\textit{b}:2], [\textit{c}:5]$>$), \textit{u}($<$[\textit{b}:3], [\textit{c}:5]$>$)\} = $ max $\{\$13, \$25, \$30\} = \$30. In other words, the overall utility of $<$$[b], [c]$$>$ in sequence $S_3$ is calculated as \$30. As shown in the above results, a sequence may have multiple utility values in a \textit{q}-sequence, which also shows the differences between HUIM and HUSPM. Thus, the utility of $<$[\textit{b}], [\textit{c}]$>$ in \textit{QSD} is:  \textit{u}($<$[\textit{b}], [\textit{c}]$>$) = \textit{u}($<$[\textit{b}], [\textit{c}]$>$, $S_{1}$) + \textit{u}($<$[\textit{b}], [\textit{c}]$>$, $S_{2}$) + \textit{u}($<$[\textit{b}], [\textit{c}]$>$, $S_{3}$) + \textit{u}($<$[\textit{b}], [\textit{c}]$>$, $S_{4}$) + \textit{u}($<$[\textit{b}], [\textit{c}]$>$, $S_{5}$) = \$27 + \$31 + \$30 + \$37 + \$35 = \$160.

\begin{definition}
	\label{Def:HUSP}
	\rm A sequence $ t $ in a \textit{QSD} is denoted as a high-utility sequential pattern (HUSP) if its overall utility in \textit{QSD} is no less than the individualized minimum utility threshold of $t$, such that: $HUSP\gets\{t|u(t)\geq MIU(t)\}$.
\end{definition}

For instance, the sequence $<$[\textit{b}], [\textit{c}]$>$ is not a HUSP since its utility in \textit{QSD} is less than \textit{MIU}($<$[\textit{b}], [\textit{c}]$>$) = \textit{min}\{$ mu(b$), $ mu(c) $\} = \textit{min}\{\$500, \$500\} = \$500. 
There have been many studies on utility mining from sequence data, which  aim at extracting a set of high-utility sequential patterns (HUSPs) with the parameter setting as a uniform minimum utility threshold. However, the addressed problem with HUSPs is not depended on the individualized minimum utility threshold. This motivates us to work on a new mining task with an efficient mining algorithm. Based on the above definitions, the problem statement of mining high-utility sequential patterns across multi-sequences with individualized thresholds can be formulated as a new utility mining problem, as described below.

\textbf{Problem Statement:} For a sequential database with varied quantities named \textit{QSD}, and a multiple unit utility table called \textit{M-table}, the problem of mining high-\textbf{\underline{U}}tility \textbf{\underline{S}}equential \textbf{\underline{P}}atterns across multi-sequences with individualized \textbf{\underline{T}}hresholds (USPT for short) is to discover the set of HUSPs where the utility of the satisfied pattern is no less than its  $ MIU $ value in $QSD$.

It is important to notice that if the minimum individualized thresholds in \textit{M-table} are all equal a uniform minimum utility threshold, the USPT problem can be considered as a traditional HUSPM algorithm. In other words, the traditional HUSPM model is a special case of the USPT problem statement.

\section{The Proposed USPT Framework}
\label{sec:4}
In this section, the baseline and improved variants for the developed USPT framework are designed, and details are given below.

\subsection{Lexicographic-Sequential Tree and Utility-Array}
As the representation of search space in the designed USPT algorithm, a lexicographic-sequential (LS)-tree is built to ensure completeness and correctness for mining HUSPs. The construction procedure can be described as follows:  the database is first scanned to find the potential 1-sequences against their potential minimum utility (\textit{PMIU}) thresholds (details are given in subsection 4.3), and the LS-tree can be constructed by the 1-sequences using the depth-first search strategy. The offspring nodes of a parent node can be combined using  \textit{I}-\textit{Concatenation} or \textit{S}-\textit{Concatenation} (details are given in subsection 4.2). From the structure of Fig. \ref{fig:tree} (constructed from the running example of Table \ref{table:db}), it can be observed that the LS-tree is an enumeration tree, which lists the set of complete sequences, and each node in the LS-tree can be used to represent a sequence. Each node in the LS-tree represents a sequence.

\begin{figure}[!htbp]	
	\centering
	\includegraphics[width=4.5in]{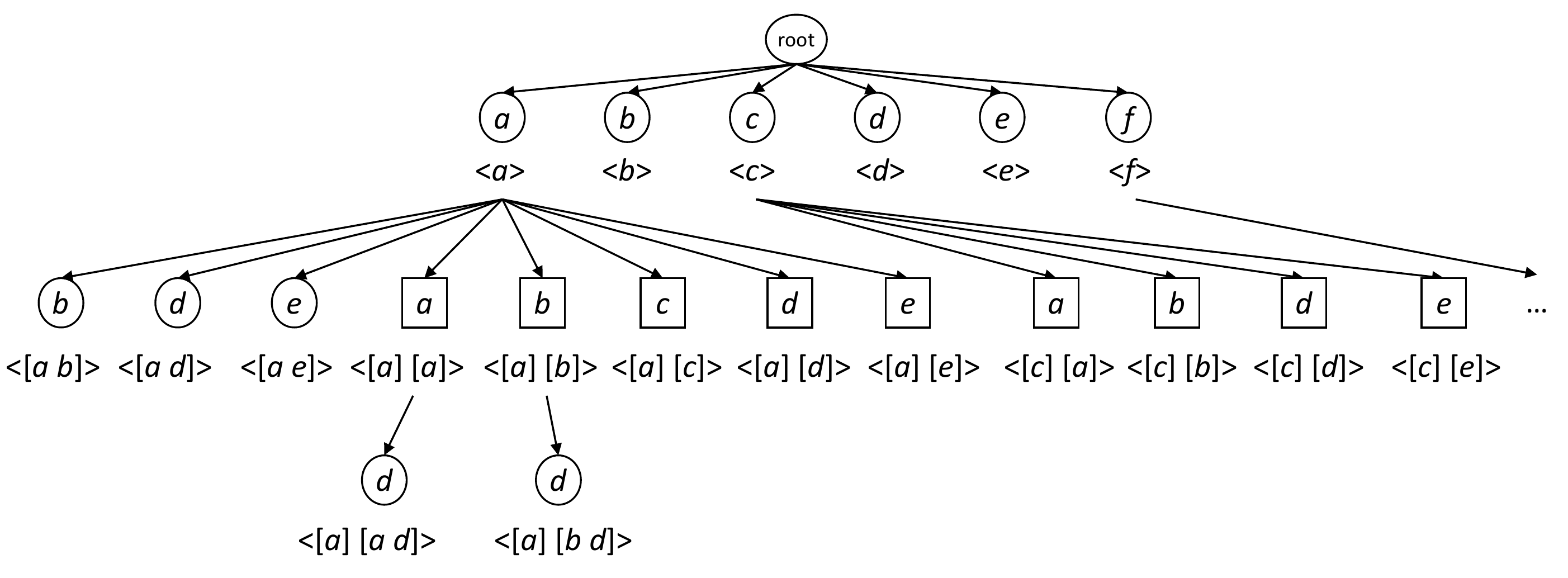}
	\caption{A lexicographic-sequential (LS)-tree \cite{lin2017high}.}
	\label{fig:tree}
\end{figure}

For instance, in Fig. \ref{fig:tree}, $ < $\textit{a}$ > $ to $ < $\textit{f}$ > $ are the 1-sequences and they are the offsprings of the root. The \textit{I}-\textit{Concatenation} sequence is represented as a circle in Fig. \ref{fig:tree} and the \textit{S}-\textit{Concatenation} sequence is represented as a square. The LS-tree also shows the number of candidates in the search space to discover the set of HUSPs.

In the pattern mining task \cite{agrawal1994fast,han2004mining,3gan2018survey}, especially in sequential pattern mining \cite{pei2001prefixspan,fournier2017survey}, there is a huge computational cost if none upper-bound of the patterns is used to prune the search space. To solve this limitation, the utility-array \cite{gan2019proum} structure is utilized to efficiently calculate the utility and position of the potential patterns. It is used to maintain the necessary information of the patterns in the processed dataset.

\begin{definition} (utility-array \cite{gan2019proum})
	\rm In a $q$-sequence database, suppose all the items in a $q$-sequence $s$ have different unique occurred  positions as \{$p_1$, $p_2$, $ \cdots $, $p_k$\}, where \{$p_1$ $< p_2$ $< \cdots < p_k$\}, and the total number of positions ($p_k$) is equal to the length of $s$.  The utility-array of a $q$-sequence $s$ = $\langle e_1, e_2, \cdots, e_n \rangle$ consists of a set of arrays, from left to right in $s$, where $e_n$ is the $n$-th element in $s$. For an item $i_j$ in each position $p_k$ in $s$, it represents an array and contains six fields: \textit{array}$_{p_k}$ = $[$\textit{eid}, \textit{item}, \textit{u}, \textit{ru}, \textit{next\_pos}, \textit{next\_eid}$]$, in which (1) \textit{eid} is the element ID of an element containing $i_j$; (2)  \textit{item} is the  name of item $i_j$; (3)  $u$ is the actual utility of $i_j$ in position $p_k$; (4)  $ru$ is the remaining utility of $i_j$ in position $p_k$; (5)  \textit{next\_pos} is the next position of $i_j$ in $s$; (6)  \textit{next\_eid} is the position of the first item in next element (\textit{eid}+1) after current element (\textit{eid}). Besides, the utility-array records the first occurred  position of each distinct item in a $q$-sequence $s$.

\end{definition}

\begin{table}[!htbp] 
	\setlength{\abovecaptionskip}{0pt}
	\setlength{\belowcaptionskip}{0pt}
	\caption{The utility-array structure of $ S_{3} $ (notation ``-" means empty)}.
	\label{table:utilityarray}
	\centering
	\begin{tabular}{|c|c|c|c|c|c|c|}
		\hline
		&  \textit{\textbf{eid}}  & \textit{\textbf{item}}  &  \textit{\textbf{u}}  &  \textit{\textbf{ru}}  &  \textit{\textbf{next\_pos}}  &  \textit{\textbf{next\_eid}} \\ \hline \hline
		\textit{array$_1$} & 1   & $a$  & \$12  &  \$84 &  3  & 3 \\ \hline
		\textit{array$_2$} & 1   & $b$  & \$10  &  \$72 &  4  & 3 \\ \hline
		\textit{array$_3$} & 2   & $a$  & \$8  &  \$64 &  -  & 6 \\ \hline
		\textit{array$_4$} & 2   & $b$  & \$15  &  \$49 &  6  & 6 \\ \hline
		\textit{array$_5$} & 2   & $c$  & \$3  &  \$46 &  7  & 6 \\ \hline
		\textit{array$_6$} & 3   & $b$  & \$20  &  \$26 &  -  & 9 \\ \hline
		\textit{array$_7$} & 3   & $c$  & \$15  &  \$11 &  -  & 9 \\ \hline
		\textit{array$_8$} & 3   & $e$  & \$8  &  \$3 &  -  & 9 \\ \hline
		\textit{array$_9$} & 4   & $d$  & \$3  &  \$0 &  -  & - \\ \hline
	\end{tabular}
\end{table}

The utility-array structure of $ S_{3} $ from  Table \ref{table:db} can be illustrated as shown in Table \ref{table:utilityarray}. For instance, in Table \ref{table:utilityarray}, the distinct items of $ S_{3} $ are (\textit{a}), (\textit{b}), (\textit{c}), (\textit{d}), and (\textit{e}), respectively. The first occurred positions of \textit{a}), (\textit{b}), (\textit{c}), (\textit{d}), and (\textit{e}) in $ S_{3} $ are 1, 1, 5, 9, and 8, respectively. Consider the item (\textit{b}) in Table \ref{table:utilityarray}, it occurs in position 2 (\textit{array$_2$}), 4 (\textit{array$_4$}), and 6 (\textit{array$_6$}), respectively. In its first occurred position 2 in $ S_{3} $, the element ID is 1,  the utility of item (\textit{b}) is \$10, and the remaining utility of item (\textit{b}) is measured as \$72, which is the remaining utility. In addition, the next position of an item (\textit{b}) in $ S_{3} $ is 4 (\textit{array$_4$}), and the position of the first item in next element ([(\textit{a}:2) (\textit{b}:3) (\textit{c}:1)]) in $ S_{3} $ is 3. Thus, the utility-array can store the necessary information (e.g, utility, position, and sequence order), which can be built during the construction of LS-tree.

In the LS-tree, the utility-array  of each node/sequence can be directly constructed from the utility-array of its prefix node, without scanning the database again. We adopt the projection mechanism to recursively construct a series of utility-arrays, as well as their extension nodes (aka super-sequences) with the same prefix in the LS-tree. For each node/sequence in the LS-tree, the transactions with the processed node/sequence are then converted into a utility-array, which is attached to the projected utility-array of the processed node. Thus, the utilities and the upper-bounds of the potential HUSPs (or candidates) can be efficiently measured.

\subsection{Two Concatenations}
In this section, two basic operations, \textit{I}-\textit{Concatenation} and \textit{S}-\textit{Concatenation}, in the LS-tree are introduced. In the LS-tree, they are used to generate the combinations of the prefix sequences. In other words, the two operations produce the offspring (also called extensions) of the processed nodes.  

\begin{definition}
	Let $ t $ and  $ i_{j} $ represent a sequence and an item, respectively. Additionally, let $<$$t \oplus i_{j}$$>$$_{I-Concatenation}$ represent the \textit{I-Concatenation} of $ t $ with $ i_{j} $, in which  $ i_{j} $ is appended to the last element of a sequence $ t $. Moreover, let $<$$t \oplus i_{j}$$>$$_{S-Concatenation}$ represent the \textit{S-Concatenation} of $ t $ with $ i_{j} $, in which $ i_{j} $ is added as a new element/itemset to the last of $ t $. Thus, the size of new sequence $<$$t \oplus i_{j}$$>$$_{S-Concatenation}$ increases by one, but the size of new sequence $<$$t \oplus i_{j}$$>$$_{I-Concatenation}$ does not change.
\end{definition}

For instance, we generate the new combinations with a sequence $ t $ = $<$[$b$], [$c$]$>$ and a new item $(a)$. The results are $<$$t \oplus a$$>$$_{I-Concatenation}$ = $<$[\textit{b}], [\textit{ca}]$>$ and $<$$t \oplus a$$>$$_{S-Concatenation}$ = $<$[\textit{b}], [\textit{c}], [\textit{a}]$>$. Obviously, the number of itemsets in $ t $ does not change after \textit{I}-\textit{Concatenation}. However, the number of itemsets in new sequence generated by \textit{S-Concatenation} increases by one. Thus, the complete candidates can be easily produced by the two operations in the search space.

As previously described, utility-array \cite{gan2019proum} is efficient for retrieving the utilities and upper-bounds of the candidate, and the actual HUSPs can be completely discovered. However, a sequence may have multiple matches in a \textit{q}-sequence, and therefore the sequence may obtain multiple utilities. It is thus necessary to obtain the accurate positions of the matches, which can be used to measure the utilities and the upper-bound values of the processed node (sequence). For convenience, we introduce the concepts for concatenation.

\begin{definition}
	\rm Let \textit{t} and \textit{s} represent a sequence and a \textit{q}-sequence, respectively, and assume that $ t \sim s $. The concatenation point of \textit{t} in \textit{s} is defined as the position of the final (last) item within each match. The first concatenation point is called the \textit{start point}. 
\end{definition}

For instance, in Table \ref{table:db}, we assume that a sequence $ t $ = $<$[\textit{b}], [\textit{c}]$>$; and it has three matches in $ S_{3} $ = $ < $[(\textit{a}:3) (\textit{b}:2)], [(\textit{a}:2) (\textit{b}:3) (\textit{c}:1)], [(\textit{b}:4) (\textit{c}:5) (\textit{e}:4)], [(\textit{d}:3)]$ > $ such that $<$[\textit{b}:2], [\textit{c}:1]$>$, $<$[\textit{b}:2],  [\textit{c}:5]$>$, and $<$[\textit{b}:3], [\textit{c}:5]$>$. Thus, the concatenation positions of $ t $ in $ S_{3} $ are 5, 7, and 7, respectively. The start point of $<$[\textit{b}], [\textit{c}]$>$ in $ S_{3} $ is thus set as 5.

Following the definition of \textit{I-Concatenation}, a new additive item is appended to the last itemset of the processed sequence. Based on \textit{I-Concatenation}, the candidate items are the items occurring in the same elements/itemsets in which the concatenation points arise. For instance, the candidate items for \textit{I-Concatenation} of $<$[\textit{b}:2], [\textit{c}:1]$>$ are \{(\textit{a}:2), (\textit{b}:3)\}. For the \textit{S-Concatenation}, the new items are appended to the last of the sequence as the new itemset. For the \textit{S}-\textit{Concatenation} operation, the items in the elements/itemsets after the start point are then considered to be the candidate items of each transaction. In the illustrated example, the start point is equal to 5, which appears in the second element in $S_3$. Thus, the items after the second element are considered as the candidate items for \textit{S-Concatenation}, such as \{(\textit{b}:4), (\textit{c}:5), (\textit{e}:4), (\textit{d}:3)\} in $S_3$. Since multiple matches of a sequence \textit{t} in a \textit{q}-sequence can be obtained, the highest utility of the items in the sequence \textit{t} is then regarded as the utility of \textit{t} for \textit{q}-sequence. 

In addition, we define the order of sequences in the proposed USPT algorithm and the LS-tree. Given two sequences, $ t_{a} $ and $ t_{b} $, then $ t_{a} \prec t_{b} $, if 1) the length of $ t_{a} $ is less than that of $ t_{b} $; 2) $ t_{a} $ is \textit{I}-\textit{Concatenation} from $ t $, while $ t_{b} $ is \textit{S}-\textit{Concatenation} from $ t $; and 3) $ t_{a} $ and $ t_{b} $ are both \textit{I}-\textit{Concatenation} or \textit{S}-\textit{Concatenation} from \textit{t}, and the additive item in $ t_{a} $ is alphabetically smaller than that in $ t_{b} $. This order of sequences is also suitable for $ q $-sequences. For example, $<$[\textit{a}]$>$ $\prec$ $<$[\textit{ab}]$>$ $\prec$ $<$[\textit{a}], [\textit{a}]$>$ $\prec$ $<$[\textit{a}], [\textit{c}]$>$. To retrieve the complete set of HUSPs, all candidates of the two operations (concatenations) are enumerated with the lexicographic order in the LS-tree.

\subsection{Upper Bounds and Pruning Strategies}
Based on the concepts of LS-tree and utility-array, by adopting the depth-first search strategy and two concatenations, the proposed USPT algorithm can successfully identify the complete HUSPs. However, this process may lead to the problem of combinatorial explosion, since a huge number of candidates may be exhaustively  explored in the LS-tree. The reason for this is that the downward closure property \cite{agrawal1994fast} (also called Apriori property) cannot be directly applied to HUSPM. Therefore, to solve the problem of combinatorial explosion for mining the HUSPs, a new property is required. Several concepts of sequences and \textit{q}-sequences are introduced below.

\begin{definition}
	\rm Assume two $ q $-sequences such that \textit{t} and $ t' $, if $ t \subseteq t' $ holds, then the extension of $ t \in t' $ is the suffix of $ t' $ after $ t$. Thus, it can be denoted as $<$$ t' $-$ t $$>_{rest}$. Let \textit{t} and \textit{s} represent a sequence and a \textit{q}-sequence, respectively. If $ t \sim s_{k} \wedge s_{k} \subseteq s $ $ (t \subseteq s) $ holds, then the superset (extension) of $ t $ in $ s $ is the rest of $ s $ after $ s_{k} $, which can be denoted as $<$$s$ - $t$$>_{rest}$, and $ s_{k} $ represents the first match of $ t $ in $ s $. 
\end{definition}

For instance, two \textit{q}-sequences such as $ s $ = $<$[\textit{b}:2], [\textit{c}:5]$>$ and $ S_{3} $ are given, where $ S_{3} $  is shown in Table \ref{table:db}. The superset of $ s $ in $ S_{3} $ is $<$$S_{3}$ - $s$$>$$_{rest}$ = $<$[(\textit{e}:4)], [(\textit{d}:3)]$>$. Given a sequence $ t $ = $<$[\textit{b}], [\textit{c}]$>$, three matches of $ t $ in $ S_{3} $ are existed, and the first one is $<$[\textit{b}:2], [\textit{c}:1]$>$. Thus, $<$$S_{3}$ - $t$$>$$_{rest}$ = $<$[(\textit{b}:4) (\textit{c}:5) (\textit{e}:4)], [(\textit{d}:3)]$>$.

\begin{definition}
	\rm The set of extension items of a sequence $ t $ in a sequential database is denoted as $ I(t)_{rest} $, and defined as:
	\begin{equation}
	I(t)_{rest} = \{i_{j}|i_{j}\in <s - t>_{rest} \wedge  t  \subseteq  s \wedge s \in QSD\}.
	\end{equation}
\end{definition}

\begin{definition}
	\label{def:swu}
	\rm The sequence-weighted utilization (\textit{SWU}) \cite{yin2012uspan}  of a sequence $ t $ in a $ QSD $ is denoted as $ SWU(t) $, and defined as:
	\begin{equation}
	SWU(t) = \sum_{s\in QSD} \{u(s)|t \subseteq s\}.
	\end{equation}
\end{definition}

From the given example, because $<$$S_{3}$ - $t$$>$$_{rest}$ = $<$[[(\textit{b}:4) (\textit{c}:5) (\textit{e}:4)], [(\textit{d}:3)]$>$ in $S_{3}$, \textit{I}($<$[\textit{b}], [\textit{c}]$>$)$_{rest}$ = \{$ b, c, e, d $\}. In Table \ref{table:db}, $ SWU $($<$\textit{b}$>$) = $ u(S_{1}) $ + $ u(S_{2}) $ + $ u(S_{3}) $ + $ u(S_{4}) $ + $ u(S_{5}) $ + $ u(S_{6}) $= \$56 + \$67 + \$94 + \$67 + \$76 + \$81 = \$441, and $ SWU $($<$\textit{f}$>$) = $ u(S_{6})$ = \$81.
Consider $t$ = $<$[\textit{b}], [\textit{c}]$>$, \textit{SWU}($<$[\textit{b}], [\textit{c}]$>$) = $ u(S_{1}) $ + $ u(S_{2}) $ + $ u(S_{3}) $ + $ u(S_{4}) $ + $ u(S_{5}) $ = \$56 + \$67 + \$94 + \$67 + \$76 = \$360.

\begin{theorem}[\textbf{\underline{S}}equence-\textbf{\underline{W}}eighted \textbf{\underline{D}}ownward \textbf{\underline{C}}losure property, SWDC \cite{yin2012uspan}]
	\label{theorem-swu}
	Given a sequential database $ QSD $ and two sequences $ t $ and $ t' $. Based on the concepts of \textit{SWU}, if $ t \subseteq t' $ holds, then we can obtain the following property as:
	\begin{equation}
		SWU(t')\leq SWU(t).
	\end{equation}
\end{theorem}

\begin{theorem}
	\label{theorem-swu-upper-bound}
	Given a sequential database $ QSD $ and a sequence $ t $, then the following equation holds:
	\begin{equation}
		u(t)\leq SWU(t).
	\end{equation}
\end{theorem}

The SWDC property and Theorem \ref{theorem-swu-upper-bound} ensure that if the $ SWU $ of a sequence $ t $ is no larger than a threshold, then the utility of $ t $ and any of its super-sequences are also less than this threshold. Therefore, numerous unpromising candidates can be pruned. To speed up the mining performance, the concept of sequence-weighted utilization (\textit{SWU}) \cite{yin2012uspan} is proposed to maintain the extended downward closure property for mining HUSPs. Thus, the search space can be reduced and some unpromising candidates can be promptly pruned.   However, the $ SWU $ of a sequence \textit{t} is usually greater than its actual utility of \textit{t}. To make the designed algorithm more efficient, the remaining utility model \cite{yin2012uspan} is introduced to estimate the lower upper-bound values of the candidates.

\begin{definition}
	\rm Let $ SEU(t) $ denote the sequence extension utility of a sequence $ t $ in a  sequential database $ QSD $. As an upper bound, the sequence extension utility (\textit{SEU}) \cite{gan2019proum}  can be defined as:
	\begin{equation}
		SEU(t) = \sum_{s\in QSD} \{u(t,s) + u(<s-t>_{rest})|t \subseteq s\}.
	\end{equation}
\end{definition}

It should be noted that \textit{u}($<$\textit{s} - \textit{t}$>_{rest}$) is the remaining utility with respect to the first match of $ t $ in $ s $, which is kept as the fourth field in the utility-array. For instance, in Table \ref{table:db}, suppose a sequence $ t $ = $<$[\textit{b}], [\textit{c}]$>$, \textit{u}($<$$S_{2}$ - $t$$>$$_{rest}$) = \textit{u}($<$(\textit{d}:3), (\textit{e}:4)$>$, $S_2)$ = \$3 + \$8 = \$11. Thus, $ SEU(t) $ = ($ u(t, S_{1})$ + \textit{u}($<$$S_{1}$ - $t$$>$$_{rest}$)) + ($ u(t, S_{2})$ + \textit{u}($<$$S_{2}$ - $t$$>$$_{rest}$)) + ($ u(t, S_{3})$ + \textit{u}($<$$S_{3}$ - $t$$>$$_{rest}$)) + ($ u(t, S_{4})$ + \textit{u}($<$$S_{4}$ - $t$$>$$_{rest}$)) + ($ u(t, S_{5})$ + \textit{u}($<$$S_{5}$ - $t$$>$$_{rest}$)) = (\$27 + \$10) + (\$31 + \$11) + (\$30 + \$46) + (\$37 + \$11) + (\$35 + \$11) = \$249, which is less than \textit{SWU}($t$)( = \$360).

\begin{theorem}
	\label{theorem-pu}
	Suppose two sequences, $ t $ and $ t' $, in \textit{QSD}. If $ t \subseteq t' $ holds, then: 
	\begin{equation}
		SEU(t') \leq SEU(t).
	\end{equation}		
\end{theorem}

\begin{theorem}
	\label{theorem-pu-upper-bound}
	Given a sequence \textit{t} in $ QSD $, then:
	\begin{equation}
		u(t)\leq SEU(t).
	\end{equation}	
\end{theorem}

\begin{proof}
	$\displaystyle u(t)$ = $\sum_{s \in QSD} \{u(t,s)|t\subseteq s\} \leq \sum_{s\in QSD} \{u(t,s) + u(<s - t>_{rest})|t \subseteq s\}$ = $SEU(t)$.
\end{proof}

Theorems \ref{theorem-pu} and \ref{theorem-pu-upper-bound} indicate that for a sequence $ t $ in sequence database, if $ SEU(t) $ is no greater than a minimum utility threshold, then the utility of $ t $ and any of its super-sequences are no greater than this minimum utility threshold. However, the upper bounds (both $ SWU $ and $ SEU $)  cannot be directly utilized in our USPT framework since each item may obtain its individualized minimum utility threshold. Thus, by using the upper bounds of $ SWU $ and $ SEU $, the USPT framework does not hold the initial downward closure property. For instance, let $ t' $ be an extension (super-sequence) of $ t $; although the $SEU$ or $ SWU $ of $ t $ is no greater than the minimum utility threshold of $ t $, $ t' $ still has the possibility for becoming a high-utility sequential pattern. 

As mentioned, the USPT framework first studies the individualized threshold for each item in utility mining. Thus, several new theorems are presented to hold the downward closure property for correctly and completely mining HUSPs.

\begin{definition}
	\rm The potential minimum utility threshold  of a sequence $ t $ in \textit{QSD} is denoted as \textit{PMIU}$(t) $. Thus, \textit{PMIU}$(t) $ is the potential least \textit{mu} threshold among $ t $ and its extensions, and defined as:
	\begin{equation}
		PMIU(t) = min\{mu(i_{j})|i_{j}\in t \vee i_{j}\in I(t)_{rest}\}.
	\end{equation}
\end{definition}

Obviously, the concept of \textit{PMIU} is different from \textit{MIU}. The reason is that \textit{PMIU} is related to a sequence $ t $ and its extensions, while \textit{MIU} is only related to a sequence $ t $. For instance, in Table \ref{table:db}, \textit{PMIU}($<$[\textit{b}], [\textit{c}]$>$) = \textit{min}\{$mu(i_{j})|i_{j}\in$$<$[\textit{b}], [\textit{c}]$>$$\vee i_{j} \in $\{\textit{b}, \textit{c}, \textit{e}, \textit{d}\}\} = \textit{min}\{\$500, \$500, \$200, \$500\} = \$200.

\begin{theorem}
	\label{theorem:all-upper-bound}
	Based on the definition of HUSP (cf. Definition \ref{Def:HUSP}),  assume a sequence $ t $, and let  \textit{UB}$(t) $ be the upper-bound on utility of $ t $. We have the theorem as: if $ UB(t) < PMIU(t) $, then $ t $ and any of its super-sequences will not be the HUSPs.
\end{theorem}

\begin{proof}
	Let $ t' $ be a extension (super-sequence) of $ t $, then $ \{i_{j} | i_{j} \in t'\} \subseteq \{i_{j} | i_{j} \in t \vee i_{j} \in I(t)_{rest}\} $ and $ I(t')_{rest} \subseteq I(t)_{rest} $ hold. It is then obtained that $ u(t') \leq UB(t') $ and $ u(t) \leq UB(t) $. Next, the following equation can be obtained:
	
	\begin{align*}
		u(t') &\leq UB(t')\\
		&< PMIU(t) = min\{mu(i_{j}) | i_{j}\in t \vee i_{j}\in I(t)_{rest}\}\\
		&\leq min\{mu(i_{j}) | i_{j}\in t' \vee i_{j}\in I(t')_{rest}\} = PMIU(t')\\		
		&\leq min\{mu(i_{j}) | i_{j}\in t'\} = MIU(t').\\
		u(t) &\leq UB(t)\\
		&< PMIU(t) = min\{mu(i_{j}) | i_{j}\in t \vee i_{j}\in I(t)_{rest}\}\\
		&\leq min\{mu(i_{j}) | i_{j}\in t\} = MIU(t).
	\end{align*}
	It is then ensured that $ t $ and $ t' $ will not be the HUSPs.
\end{proof}

Based on Theorem \ref{theorem:all-upper-bound}, if the upper-bound of a sequence $t $ is no greater than the \textit{PMIU} of $ t $, $ t $ and any of its super-sequences will not be the HUSPs. The reason for this is that the upper-bound is the maximum utility of a sequence in all possible combinations, which can be either the $SWU$ or $SEU$ of \textit{t}. In addition, the \textit{PMIU} is the least and minimum threshold of the sequence. According to the above theorems, the correctness and completeness of the designed USPT algorithm can be guaranteed for mining the HUSPs.

\textbf{\textit{PMIU}-based \textit{SEU} Strategy:} The candidate sequences whose \textit{SEU} values are no greater than their \textit{PMIU} values are discarded from the candidate set so that their child nodes (or called extensions) are not generated and explored in the LS-tree.

The candidate items for the generation, e.g., \textit{I-Concatenation} or \textit{S-Concatenation} of $ t $, are measured by their $SWU$ values to reduce the search space. If the $SWU$ value is less than the \textit{PMIU} value, the candidate item is removed since the super-sequences of $ t $ concatenated with this candidate item can not be HUSPs. The set of \textit{I-Concatenation} or \textit{S-Concatenation} sequences $ t' $ ($ t' $ is the concatenation of $ t $ and a candidate item) are determined by their $SEU$ values to check whether they are the candidate sequences. If the $SEU$ of a sequence $ t' $ is less than its \textit{PMIU}, then $ t' $ can be safely discarded since $ t' $ and any of the super-sequences of $ t' $ cannot be HUSPs.

The designed baseline USPT algorithm can be used to find the HUSPs correctly and completely. However, the search space of the baseline USPT algorithm is still huge since it utilizes the $ SWU $ and $ SEU $ over-estimated values, which is much larger than the actual utility of the pattern. Three strategies are developed here to speed up mining performance by promptly pruning the unpromising candidates in the search space. We first utilize a tighter upper-bound on utility for mining HUSPs \cite{wang2016efficiently}, and the details are given below.

\begin{definition}
	\rm The prefix extension utility (\textit{PEU}) \cite{wang2016efficiently} of a sequence $ t $ in a \textit{q}-sequence $ s $ is defined as \textit{PEU}$(t, s)$ = $max\{u(s_{k}) + u(\textnormal{$ < $}s \textnormal{$ - $} s_{k}\textnormal{$ > $}_{rest})|t \sim s_{k} \wedge s_{k} \subseteq s\}$. Note that here $u(\textnormal{$ < $}s \textnormal{$ - $} s_{k}\textnormal{$ > $}_{rest})$ is the $s_{k}$ match but not the first match of $ t $ in $ s $.
\end{definition}

For example, assume a sequence \textit{t} = $<$[\textit{b}], [\textit{c}]$>$, \textit{t} has 3 matches in $ S_{2} $, which are $<$[\textit{b}:1], [\textit{c}:3]$>$, $<$[\textit{b}:1], [\textit{c}:2]$>$ and $<$[\textit{b}:5], [\textit{c}:2]$>$. It can be obtained that $u$($<$$S_{2}$ - $<$[\textit{b}:1], [\textit{c}:3]$>>_{rest}$) = $u$($<$[(\textit{e}:2)], [(\textit{a}:1) (\textit{c}:2) (\textit{d}:3) (\textit{e}:4)]$>$) = \$25, $u$($<$$S_{2}$ - $<$[\textit{b}:1], [\textit{c}:2]$>>_{rest})$ = $u$($<$[(\textit{d}:3) (\textit{e}:4)]$>$) = \$11, and $u$($<$$S_{2}$ - $<$[\textit{b}:5], [\textit{c}:2]$>>_{rest}$) = $u$($<$[(\textit{d}:3) (\textit{e}:4)]$>$) = \$11. The utilities of 3 matches are \$14, \$11, and \$31, respectively. Thus, \textit{PEU}($<$[\textit{a}], [\textit{b}]$>$, $S_{2}$) = \textit{max}\{\$14 + \$25, \$11 + \$11, \$31 + \$11\} = \$42. Obviously, the concept of \textit{PEU} is different from \textit{SEU}.

\begin{definition}
	The prefix extension utility (\textit{PEU}) \cite{wang2016efficiently} of a sequence $ t $ in $ QSD $ is denoted as \textit{PEU}$(t) $ and defined as \textit{PEU}$(t) = \sum_{s\in QSD}\{PEU(t, s)|t\subseteq s\}$. Thus, \textit{PEU} is the maximal utility of an extension in a sequence or database.   
\end{definition}

For instance, in Table \ref{table:db}, it is assumed a sequence \textit{t} = $<$[\textit{b}], [\textit{c}]$>$, $ PEU(t) $ is measured as $ SEU(t, S_1) $ + $ SEU(t, S_2) $ + $ SEU(t, S_3) $ = $ SEU(t, S_4) $ + $ SEU(t, S_5) $ = \$37 + \$42 + \$59 + \$48 + \$46 = \$232, which is smaller than $SEU$(\textit{t}) = \$249.

\begin{theorem}
	\label{theorem:peu}
	Assume two sequences \textit{t} and $ t'$ in \textit{QSD}. If \textit{t} $ \subseteq $ $ t'$ holds, then:
	\begin{equation}
	PEU(t') \leq PEU(t).	
	\end{equation}
\end{theorem}

\begin{proof}
	Detailed proof can be referred to \cite{wang2016efficiently}. 
\end{proof}

Thus, Theorem \ref{theorem:peu} shows that if $PEU$  of a sequence $ t $ is no greater  than a threshold, then the $PEU$ values of its super-sequences are no greater than this threshold. The $PEU$ concept holds the downward closure property.

\begin{theorem}
	\label{theorem:upper-bound-meu}
	Suppose a sequence $ t $ in $ QSD $, then: 
	\begin{equation}
		u(t)\leq PEU(t).
	\end{equation}
\end{theorem}

\begin{proof}
	Detailed proof can be referred to \cite{wang2016efficiently}. 
\end{proof}

From Theorems \ref{theorem:all-upper-bound}, \ref{theorem:peu}, and \ref{theorem:upper-bound-meu}, the completeness and the correctness of the mined HUSPs can be ensured. If the \textit{PEU} of a sequence $ t $ is no greater than the \textit{PMIU} of $ t $, then the utility of $ t $ is no greater than the $MIU$ value of \textit{t}; the utilities of its extensions/super-sequences are no greater than their $MIU$ values ($ t $ and the super-sequences of $ t $ are not HUSPs).

\begin{theorem}
	\label{three-upper-bounds}
	Suppose a sequence \textit{t} in \textit{QSD}, then:
	\begin{equation}
		PEU(t)\leq SEU(t)\leq SWU(t).
	\end{equation}
\end{theorem}

Based on Theorem \ref{three-upper-bounds}, it can be seen that \textit{PEU} holds a tighter upper-bound value compared to $SEU$ and $SWU$. Thus, the designed algorithm can greatly reduce the unpromising candidates based on the \textit{PEU} model than that of the $SEU$ and $SWU$ models. The \textit{PEU} model can be used to  estimate the utility values for a candidate sequence and its super-sequences. Therefore, if the \textit{PEU} of a sequence $ t $ is no greater than the \textit{PMIU} of $ t $, then $ t $ and any of its super-sequences will not be considered as the HUSPs.  

\textbf{\textit{PMIU}-based \textit{PEU} Strategy:} The candidate sequences whose \textit{PEU} values are no greater than their \textit{PMIU} values are discarded from the candidate set so that their child nodes (or called extensions) are not generated and explored in the LS-tree. 

A huge number of candidates are generated using \textit{I-Concatenation} and \textit{S-Concatenation} for a processed sequence \textit{t}. Based on the \textit{PMIU} concept in the addressed problem, we further propose the \textit{PMIU}-based  pruning strategy (PUK), which can also be used to reduce the search space by promptly removing unpromising candidates.

\begin{theorem}
	\label{theorem:las}
	For a sequence $ t $ in $ QSD $, assume an item $ i_{j} $ is an \textit{I-Concatenation} or \textit{S-Concatenation} candidate item of $ t $. Two situations can be considered to produce the super-sequences/extensions as follows:
	
	Case 1: the maximal utility of $<$$t\oplus i_{j}$$>$$_{I-Concatenation}$ is always less than or equal to the upper bound on utility (e.g., \textit{SEU}, \textit{PEU}) of $t$;
	
	Case 2: the maximal utility of $<$$t$$\oplus i_{j}$$>$$_{S-Concatenation}$ is always less than or equal to the upper bound on utility (e.g., \textit{SEU}, \textit{PEU}) of $t$.
\end{theorem}

\begin{proof}
	For Case 1, let $ t' $ = $<$$t \oplus i_{j}$$>$$_{I-Concatenation}$. Since $PEU(t')$ = $ \sum_{s\in QSD} \{PEU(t', s)|t' \subseteq s\} $ and $ PEU(t', s)\leq PEU(t, s) $. For convenience, as proven in Theorem \ref{theorem:peu}, $ PEU(t')\leq PEU(t) $. According to Theorem \ref{theorem:upper-bound-meu}, $ u(t') \leq PEU(t')$ holds. Thus, $ u(t') \leq PEU(t') = \sum_{s\in QSD} \{PEU(t', s) |t' \subseteq s \} \leq \sum_{s\in QSD} \{PEU(t, s)|t \subseteq s \} $, we have $u(t') \leq PEU(t)$. Case 2 can be verified in the same way.  
\end{proof}

Based on the \textit{PMIU}-based \textit{SEU} strategy and \textit{PEU} strategy, we further propose the general strategy to prune the unpromising $k$-sequences during spanning the LS-tree.

\textbf{\textit{PMIU}-based pruning unpromising $k$-sequences  strategy (PUK strategy):} 
Assume a sequence $ t $ in $ QSD $, then two situations can be considered as follows: 

Case 1: if $ i_{j} $ is a \textit{I-Concatenation} candidate item for $ t $ and the upper bound on utility (e.g., \textit{SEU}, \textit{PEU}) of $t$ is no greater than \textit{PMIU}($<$$t \oplus i_{j}$$>$$_{I-Concatenation}) $, then $ i_{j} $ can be removed from the set of candidate items for \textit{I-Concatenation} of $ t $;

Case 2: if $ i_{j} $ is a \textit{S-Concatenation} candidate item for $ t $ and the upper bound on utility (e.g., \textit{SEU}, \textit{PEU}) of $t$  is no greater  than \textit{PMIU}($<$$t \oplus i_{j}$$>$$_{S-Concatenation}) $, then $ i_{j} $ can be removed from the set of candidate items for \textit{S-Concatenation} of $ t $.

The \textit{PMIU}-based PUK strategy is used to promptly prune unpromising candidates for the \textit{I-Concatenation} and \textit{S-Concatenation} of $ t $ without evaluating the \textit{PEU} values for $<$$t \oplus i_{j}$$>$$_{I-Concatenation}$ and $<$$t$ $\oplus i_{j}$$>$$_{S-Concatenation}$, where $ t $ is a sequence in a dataset. The reason for this is that the upper-bound value can be easily evaluated from the utility-arrays of $ t $. Thus, the designed algorithm can reduce the computational cost with a small set of candidates for either \textit{I-Concatenation} or \textit{S-Concatenation} of $ t $.

It is important to notice that \textit{PEU} provides a tighter upper-bound than \textit{SEU}, thus the search space for mining the HUSPs can be greatly reduce. In addition, more unpromising candidates can be reduced in the search space with the developed \textit{PMIU}-based pruning strategies. For example, the candidates for \textit{I-Concatenation} or \textit{S-Concatenation} of the processed sequence $ t $ can be pruned by \textit{SEU} and PUK strategy and the utility-arrays of $ t $ will be re-calculated. After that, the candidates for \textit{I-Concatenation} and \textit{S-Concatenation} of the processed $ t $ will be measured by PUK strategy instead of their $SWU$ values.

\subsection{The Designed USPT Algorithm}

Based on the above properties of upper bounds on utility and utility-array, the main procedures of the designed USPT algorithm are introduced in Algorithm \ref{procedureOfUSPT} and Algorithm \ref{Span-LQStree}. Algorithm \ref{procedureOfUSPT} shows that the whole process of USPT has three input parameters as follows: 1) \textit{QSD}; 2) \textit{utable}; and 3) \textit{M-table}. The USPT algorithm first scans  $ QSD $ once to construct utility-array of each $s\in QSD$, denoted the set as $D.ua$ (Line 1).  After that, USPT processes each 1-sequence $ i_{j}\in QSD $ one-by-one with the lexicographic order. Based on the constructed $D.ua$, it first constructs the projected utility-array $(D.ua)|_{<i_{j}>}$ of $ i_{j} \in QSD $ to store the information (e.g., utility,  position, and sequence order) (Line 4). Then the actual overall utility, $SWU$, \textit{PMIU}, and $MIU$ of each 1-sequence are respectively calculated from the projected $(D.ua)|_{<i_{j}>}$ (Line 5). USPT also updates the remaining utilities in $D.ua$ by removing the unpromising items $SWU$($<$$i_{j}$$>$) $ < $ \textit{PMIU}($<$$i_{j}$$>$) from $D.ua$ (Line 6, adopts the \textit{SWU} strategy). For each processed 1-sequence, it is discovered as a candidate item if its $SWU$ value is no less than the \textit{PMIU} value (Lines 7 to 12). If the utility of a 1-sequence is no less than the \textit{MIU} value, then this 1-sequence is outputted as a HUSP (Lines 8-10). After that, the candidate pattern that meets the condition as $SWU$($<$$i_{j}$$>$) $ \geq $ \textit{PMIU}($<$$i_{j}$$>$) has become the \textit{prefix} of \textbf{Span-Search} function for mining \textit{prefix}-based HUSPs (Line 11).

\begin{algorithm}
	\LinesNumbered
	\caption{USPT}
	\label{procedureOfUSPT}
	\KwIn {\textit{QSD}, a sequential database with quantitative values;  \textit{utable}, a table with the unit utility of each item; \textit{M-table}, a table with the threshold of each item.} 
	\KwOut {The set of \textit{HUSPs}.}

    initialize $D.ua$ = $ \emptyset$, \textit{HUSPs} $\gets \varnothing $\;
	scan $ QSD $ once to construct utility-array of each $s\in QSD$, denoted the set as $D.ua$\;
	
	\For {each 1-sequence $ i_{j}\in QSD $}
	 {
	   construct utility-array of $<$$i_{j}$$>$ as $(D.ua)|_{<i_{j}>}$ based on $D.ua$\;
		calculate $SWU$($<$$i_{j}$$>$), $u$($<$$i_{j}$$>$), \textit{PMIU}($<$$i_{j}$$>$), and $MIU$($<$$i_{j}$$>$)\;
		update the remaining utilities in $D.ua$ by removing the unpromising items (\underline{the \textit{SWU} strategy})\;
		
		\If {$SWU$($<$$i_{j}$$>$) $ \geq $ \textit{PMIU}($<$$i_{j}$$>$) (\underline{the SWU strategy})}
	    {
		\If {$u$($<$$i_{j}$$>$) $\geq$ $MIU$($<$$i_{j}$$>$)}
		   {
		      put $<$$i_{j}$$>$ into \textit{HUSPs}\;
			}
	         call \textbf{Span-Search}($ < $$ i_{j} $$ > $, $(D.ua)|_{<i_{j}>}$)\;
			}
	}
	\Return \textit{HUSPs}
\end{algorithm}

As presented at Algorithm \ref{Span-LQStree}, the \textbf{Span-Search} procedure utilizes a depth-first search to span all possible sequences with the lexicographic order. The enumerated sequences are produced by two concatenations: \textit{I-Concatenation} and \textit{S-Concatenation}, respectively. This function first scans the projected $(D.ua)|_t$ once to obtain \textit{iItem} (the set of candidate items for \textit{I-Concatenation}) and \textit{sItem} (the set of candidate items for \textit{S-Concatenation}) (Lines 1-4). At the same time,  the $PEU$ values of the candidate items are calculated from $(D.ua)|_t$ (Line 5). Then, this procedure removes the unpromising items from two candidate sets (\textit{iItem}  and \textit{sItem}) that avoid generating the unpromising \textit{I-Concatenation} and \textit{S-Concatenation} for the prefix sequence $t$ (Lines 6-7). Based on Theorems \ref{theorem-swu} and \ref{theorem-swu-upper-bound}, these unpromising candidate item $ i_{j} $ can be safely discarded. After that, each candidate in the updated sets of \textit{iItem} and \textit{sItem} is determined one-by-one.

Details of the \textbf{Span-Search} operation are given below (Lines 8-18). For each item $ i \in iItem$, USPT generates a new sequence $ t' \leftarrow $ $I$-\textit{Concatenation}$(t, i) $, and constructs the projected utility-array $(D.ua)|_{t'}$ based on $(D.ua)|_{t}$ (Lines 9-10). Note that $(D.ua)|_{t}$ is the utility-array of the prefix $t$, thus USPT can easily construct the new projected utility-array for extension node/sequence in LS-tree. The $PEU$, \textit{PMIU}, utility, and $MIU$ of new sequence $ t'$ are then calculated from $(D.ua)|_{t'}$ (Line 11). If the utility of $ t'$ is no less than the $MIU$ of $ t'$, then $ t'$ is identified as a HUSP (Lines 13-15). If the $PEU$ of $ t'$ is no less than the \textit{PMIU} of $ t'$, then the \textbf{Span-Search} procedure is performed to discover HUSPs which are based on the prefix sequence $ t'$ (Line 16). The \textbf{Span-Search} procedure is recursively called and terminated if no candidates are generated. For each item $ i \in sItem$, it is processed in the same way (Lines 19-29). Finally, the \textbf{Span-Search} procedure returns the set of discovered HUSPs w.r.t. prefix $t$. At the end, the final complete set of HUSPs would be discovered by the designed USPT algorithm.

\begin{algorithm}
	\LinesNumbered
	\caption{Span-Search}
	\label{Span-LQStree}
	\KwIn {\textit{t}: a sequence as prefix; $(D.ua)|_t$: the projected utility-array of $t$.} 
	\KwOut {\textit{HUSPs}: the set of high-utility sequential patterns with prefix $t$.}

		 initialize $iItem$ = $ \emptyset$ and $sItem$ = $ \emptyset$\;
		  scan the projected utility-array $(D.ua)|_t$ once to: \\
		1) put $I$-Concatenation items of $t$ into \textit{iItem}; \\
		2) put $S$-Concatenation items of $t$ into \textit{sItem}; \\
		3) calculate the \textit{PEU} values of these items form $(D.ua)|_t$\;
		
		 remove unpromising items $i_j \in iItem$ that have  \textit{PEU}$(i_j) < PMIU(t)$  (\underline{the PUK strategy})\;
		 remove unpromising items $i_j \in sItem$ that have  \textit{PEU}$(i_j) < PMIU(t)$  (\underline{the PUK strategy})\;
		
		\For {each item $ i \in iItem$}
		{
		  generate new sequence $ t' \leftarrow $ $I$-\textit{Concatenation}$(t, i) $\;
		  construct the projected utility-array $(D.ua)|_{t'}$ based on $(D.ua)|_{t}$\;
		  	determine $ u $($ t'$), $ PEU $($ t'$), \textit{PMIU}($ t'$), and $ MIU $($ t'$)\;

		\If{ \textit{PEU}$(t') \geq PMIU(t')$ (\underline{the PUK strategy})} 
		{
			
			\If{$ u(t') \geq MIU(t')$}
			{
				 output $t'$ into \textit{HUSPs}\;
			}
			 call \textbf{Span-Search}$(t', (D.ua)|_{t'})$\;
		}

		}

		\For{each item $ i \in sItem$}
		{
			generate new sequence $ t' \leftarrow $ $S$-\textit{Concatenation}$(t, i) $\;
			construct the projected utility-array $(D.ua)|_{t'}$\;
			
			determine $ u $($ t'$), $ PEU $($ t'$), \textit{PMIU}($ t'$), and $ MIU $($ t'$)\;

			\If{ \textit{PEU}$(t') \geq PMIU(t')$ (\underline{the PUK strategy})}
			{
				
				\If{$ u(t') \geq MIU(t')$}
				{
					output $t'$ into \textit{HUSPs}\;
				}
				call \textbf{Span-Search}$(t', (D.ua)|_{t'})$\;
			}

		}
	
		\Return \textit{HUSPs}

\end{algorithm}

\section{An Example of the Proposed Algorithm}
\label{sec:5}
In this section, a running example is given to describe the designed USPT and its variations step-by-step. Assume that the individualized minimum utility thresholds of all items in Tables \ref{table:db} and \ref{table:profit} are respectively assigned in an \textit{M-table}, as follows: \textit{M-table} = \{$ mu(a) $, $ mu(b) $, $ mu(c) $, $ mu(d) $, $ mu(e) $, $ mu(f) $\} = \{\$500, \$500, \$500, \$500, \$200, \$70\}. Note that these values can be adjusted based on the user's preference and prior knowledge. 

To remove the unpromising items from the database, the USPT method first evaluates the $SWU$ value of each item $i_{j}$ in \textit{QSD} against $ min\{mu(i_{j})|i_{j} \in QSD\}$. The unsatisfied items are then removed from \textit{QSD} and the size of \textit{QSD} can be refined. After that, the USPT algorithm scans the refined database once to calculate the transaction utility, and the results are $\{$$ u(S_{1})$: \$56, $ u(S_{2})$: \$67, $ u(S_{3})$: \$94, $ u(S_{4})$: \$67, $ u(S_{5})$: \$76, $ u(S_{6})$: \$81$\}$. The algorithm also generates the utility-arrays for each transaction in \textit{QSD}. The projected database for each 1-sequence is built based on these utility-arrays. The $SWU$, \textit{PMIU}, $MIU$, and actual utility  of each 1-sequence are calculated from the projected database, respectively. The obtained results are shown in Table \ref{table:example1}.

\begin{table}[!htbp] 	
	\setlength{\abovecaptionskip}{0pt}
	\setlength{\belowcaptionskip}{0pt} 
	\caption{The obtained information of 1-sequences in $ QSD $.}
	\centering
	\begin{tabular}{|c|c|c|c|c|c|c|}
		\hline
		\textbf{1-sequence} & \textit{a} & \textit{b} & \textit{c} & \textit{d} & \textit{e} &  \textit{f} \\ \hline
		\textit{SWU}($ < $$ i_{j} $$ > $)  & \$360 & \$441 & \$441 &  \$360 & \$441 &  \$81 \\ \hline
		\textit{PMIU}($ < $$ i_{j} $$ > $) & \$200 & \$200 & \$200 &  \$200 & \$200 &  \$70 \\ \hline
		\textit{u}($ < $$ i_{j} $$ > $)    & \$48  & \$130 & \$72 & \$17 &  \$48    & \$24 \\ \hline
		\textit{MIU}($ < $$ i_{j} $$ > $)  & \$500 & \$500 & \$500 & \$500 &  \$200  & \$70 \\ \hline
	\end{tabular}
	\label{table:example1}
\end{table}

In this example, the 1-sequence $<$\textit{a}$>$ is first processed with the lexicographic order. Since $ SWU $($<$\textit{a}$>$) > \textit{PMIU}($<$\textit{a}$>$) but $ u $($<$\textit{a}$>$) $ < $ $ MIU $($<$\textit{a}$>$), $<$\textit{a}$>$ is not a HUSP, but the super-sequences of $<$\textit{a}$>$ should be explored. Thus, $<$\textit{a}$>$ becomes the prefix pattern for the later mining processes. The designed \textbf{Span-Search} function is then performed on $<$\textit{a}$>$ to generate its super-sequences. USPT scans the projected utility-array $(D.ua)|_{a}$ once to: 1) put $I$-Concatenation items of $t$ into \textit{iItem}; 2) put $S$-Concatenation items of $t$ into \textit{sItem}; and 3) calculate the \textit{PEU} values of these items form $(D.ua)|_{a}$. For 1-sequence $<$\textit{a}$>$, the candidates are updated as \textit{iItem} = \{$b, c, d, e$\} for $I$-\textit{Concatenation} of $<$\textit{a}$>$, and \textit{sItem} = \{$a, b, c, d, e$\} for $S$-\textit{Concatenation} of $<$\textit{a}$>$. In this example, the candidate item $ \{a\} $, appearing in $ S_{3} $ and $ S_{5} $, is different from the prefix $<$\textit{a}$>$, and it can be used for the \textit{S-Concatenation} of $<$\textit{a}$>$. The \textit{PEU} values of the candidate items are respectively calculated as \{\textit{PEU}($a$): \$312, \textit{PEU}($b$): \$359, \textit{PEU}($c$): \$218, \textit{PEU}($d$): \$156, \textit{PEU}($e$): \$104, \textit{PEU}($f$): \$81\}, and those unpromising items $d$, $e$, and $f$ are removed from the sets of \textit{iItem} and \textit{sItem} since their \textit{PEU} values are less than \textit{PMIU}($<$\textit{a}$>$) (= 200). Thus, only \textit{iItem} = \{$a, b, c$\} and \textit{sItem} = \{$a, b, c$\} are used to span the LS-tree with \textit{I-Concatenation} and \textit{S-Concatenation}, respectively.

After that, the USPT framework first generates new sequence $t'$ = $<$\textit{a, a}$>$ and constructs the projected utility-array $(D.ua)|_{t'}$. However, this $(D.ua)|_{t'}$ is empty since $<$\textit{a, a}$>$ does not occurred in Table \ref{table:db}. For the next new sequence $<$\textit{a, b}$>$, its projected utility-array is then constructed. This utility-array is not empty and the utility values of $<$\textit{a, b}$>$ for the \textit{I-Concatenation} of $<$\textit{a}$>$ are \{\textit{u}($<$\textit{a, b}$>$): \$55, \textit{PEU}($<$\textit{a, b}$>$): \$170, \textit{PMIU}($<$\textit{a, b}$>$): \$500, \textit{MIU}($<$\textit{a, b}$>$): \$500\}. Based on the \textit{PEU}-based PUK strategy, the candidate sequence $<$\textit{a, b}$>$ and its super-sequences could not be the HUSPs. Therefore, USPT stops calling the  \textbf{Span-Search} function and continues to  determinate the new candidate sequence $<$\textit{a, c}$>$.

After checking all candidates for \textit{I-Concatenation} with prefix $<$$a$$>$, the USPT algorithm begins the similar projection-based search for \textit{S-Concatenation} with prefix $<$$a$$>$. Three sequences $<$[$a$], [$a$]$>$, $<$[$a$], [$b$]$>$, and $<$[$a$], [$c$]$>$, are then performed to calculate their \textit{PEU}, \textit{PMIU}, utility, and $ MIU $ values by constructing the projected utility-array $(D.ua)|_{t'}$, based on $(D.ua)|_{a}$. After that, the \textbf{Span-Search} approach is continuously performed on next 1-sequence $<$$b$$>$ to discover HUSPs, using the similar operations on 1-sequence $<$$a$$>$. This procedure is then processed repeatedly until no candidates are generated. In this illustrated example, the final discovered HUSPs are shown in Table \ref{table:husps}.


\begin{table}[!htbp]
	\setlength{\abovecaptionskip}{0pt}
	\setlength{\belowcaptionskip}{0pt} 
	\caption{The discovered HUSPs.}
	\centering
	\begin{tabular}{|c|c|c|}
		\hline
		$ HUSP $ & $ utility $ & \textit{PMIU}  \\ \hline
		$<$[\textit{b}], [\textit{c, e}]$>$ & \$200 & \$200    \\ \hline
		$<$[\textit{f}], [\textit{b, c}], [\textit{b}]$>$ & \$73  & \$70   \\ \hline
		$<$[\textit{f}], [\textit{b, c}], [\textit{b, e}]$>$ & \$81  & \$70   \\ \hline
		$<$[\textit{f}], [\textit{b}], [\textit{b, e}]$>$ & \$72  & \$70   \\ \hline		
	\end{tabular}
	\label{table:husps}
\end{table}

From Table \ref{table:husps}, if the user focuses on the items (\textit{f}) and (\textit{e}), then the \textit{mu} values for (\textit{e}) and (\textit{f}) can be set lower. Even though the utilities of these sequences are not high, the designed algorithm can identify the sequences containing interesting items.

\section{Experimental Results}
\label{sec:6}

Several experiments were performed to evaluate the efficiency and effectiveness of the designed USPT framework. The compared algorithms include the baseline algorithm and its variations with several developed pruning strategies. For convenience, the baseline is called USPT1 (without Line 6 in Algorithm \ref{procedureOfUSPT} and Lines 6-7 in Algorithm \ref{Span-LQStree}), the variation without Lines 6-7 in Algorithm \ref{Span-LQStree} is called USPT2, and the variation with all three pruning strategies as shown in Algorithm \ref{Span-LQStree} is called USPT. This is the first work to handle the issue of discovering the HUSPs with individualized thresholds. Therefore, the baseline algorithm with its variations are compared to evaluate the mining efficiency, in terms of execution time, pattern analysis, memory consumption, and scalability test.

Besides, it is important to notice that the addressed USPT framework equals to high-utility sequential pattern mining (HUSPM) whereas all the individualized thresholds in \textit{M-table} equal to the unified minimum utility threshold. Therefore, a well-known HUS-Span \cite{wang2016efficiently} algorithm is compared with USPT to evaluate their effectiveness. A model introduced in \cite{2lin2016efficient} was used to set the varied thresholds for the items in datasets. The function of our model was set as: $mu(i_{j}) = max\{\beta\times u(i_{j}), LMU\}$. In this function, $ LMU $ represents the least minimum utility threshold, which can be adjusted by user's preference; $ u(i_{j}) $ shows the utility of each item $ i_{j} $ in the dataset; and $\beta $ states the constant factor, which is for adjusting the $ mu $ value of items.

\subsection{Experimental Environment and Datasets}

The all compared algorithms were implemented in Java language, and the experiments were conducted on a PC ThinkPad T470p. An Intel(R) Core(TM) i7-7700HQ CPU was used with 32 GB of main memory. The running operating system was set as 64-bit Microsoft Windows 10.

In the experiments, one synthetic dataset \cite{agrawal1994dataset} and five real-life datasets\footnote{\url{http://www.philippe-fournier-viger.com/spmf/index.php}} were used to evaluate the performance of the designed USPT algorithm. The characteristics of these datasets are described in Table \ref{table:paras} and Table \ref{tab:datasets}, and details are as follows:

\begin{table}[!htbp] 	
	\setlength{\abovecaptionskip}{1pt}
	\setlength{\belowcaptionskip}{1pt} 
	\caption{Parameters of used datasets.}
	\label{table:paras}
	\centering
	\begin{tabular}{|c|l|}
		\hline
		$\mathbf{\#|D|}$ &	Number of sequences \\ \hline
		$\mathbf{\#|I|}$ &	Number of distinct items     \\ \hline
		\textbf{\#S} &  Length of a sequence $s$   \\ \hline
		\textbf{\#Seq} &  Number of elements per sequence   \\ \hline
		\textbf{\#Ele} & Average number of items w.r.t. per element/itemset                \\ \hline 
	\end{tabular}
	
\end{table}

\begin{table*}[h]
	\setlength{\abovecaptionskip}{1pt}
	\setlength{\belowcaptionskip}{1pt} 
	\caption{Dataset features}
	\label{tab:datasets}
	
		\begin{tabular}{| l | c | c | c | c | c | c | l |}
			\hline
			\textbf{Dataset} & \textbf{\#$|D|$} & \textbf{\#$|I|$} & \textbf{avg(\#S)} & \textbf{max(\#S)} & \textbf{avg(\#Seq)} & \textbf{ave(\#Ele)}  \\
			\hline \hline
			\multirow{6}{*}{} 
			Sign & 730  & 267  & 52 & 94 &  51.99 & 1.0 \\

			Bible & 36,369 & 13,905 & 21.64 & 100 & 17.85  & 1.0 \\
			
			SynDataset-160k & 159,501 & 7,609 & 6.19 & 20 & 26.64 & 4.32 \\
					
			Leviathan & 5,834 & 9,025 & 33.81 & 100 & 26.34  & 1.0 \\
			
			MSNBC & 31,790 & 17 & 13.1 & 100 & 5.33 & 1.0 \\
					
			yoochoose-buys & 234,300 & 16,004 & 1.13 & 21 & 2.11 & 1.97 \\
			\hline			
	\end{tabular}

\end{table*}

\begin{itemize}
	\item \textbf{Sign}: a real-life sign language utterance dataset, which was created by the National Center for Sign Language and Gesture Resources at Boston University. It has 730 sequences with total 267 distinct items. Each utterance in this dataset is associated with a segment of video with a detailed transcription. 
	
	\item \textbf{Bible}: a real-life dataset of the conversion of the Bible into a set of sequences and each word is an item. 
	
	\item \textbf{C8S6T4I3D$ | $X$ | $K}: a synthetic sequential dataset, which was generated by IBM Quest Dataset Generator \cite{agrawal1994dataset}, and used to evaluate the scalability of USPT. Note that \textbf{SynDataset-160K} is a subset of \textbf{C8S6T4I3D$ | $X$ | $K}, and only contains 160,000 sequences.
		
	\item \textbf{Leviathan}: a conversion of the 1651 novel, Leviathan, by Thomas Hobbes. Each word is an item in the sequences.

	\item \textbf{MSNBC}: a real-life click-stream dataset, which was accessed from the UCI repository, and all the shortest sequences are removed.  It contains 31,790 sequences and 17 distinct items.
	
	\item \textbf{yoochoose-buys}: a real-life e-commercial dataset, which was collected from YOOCHOOSE GmbH\footnote{\url{https://recsys.acm.org/recsys15/challenge/}}. We preprocess this dataset by removing some noise data.

\end{itemize}

Since these sequence datasets  except for yoochoose-buys do not contain the utility values, a well-known simulation model \cite{liu2005two} was utilized to generate the internal and external utility values of the items in these  datasets. For the quantity value of each item, a random function was used to generate its value in a range from 1 to 5. A log-normal distribution was used to randomly assign the unit utility value of the items in a range from 1 to 1000. All the test datasets, except for \textbf{C8S6T4I3D$ | $X$ | $K}, can be accessed from the SPMF website\footnote{\url{http://www.philippe-fournier-viger.com/spmf/index.php}}.

\subsection{Runtime Analysis}
In our experiments, each runtime of the designed approach with different pruning strategies was first evaluated to analyze the efficiency of USPT algorithm. Results under various \textit{LMUs} with a fixed $\beta$ of runtime comparison are shown in Fig. \ref{fig:runtime1}. Moreover, results under various $\beta$ with a fixed \textit{LMU} are shown in Fig. \ref{fig:runtime2}.

\begin{figure}[!htbp]
	\setlength{\abovecaptionskip}{0pt}
	\setlength{\belowcaptionskip}{0pt}	
	\centering
	\includegraphics[trim=40 10 30 5,clip,scale=0.45]{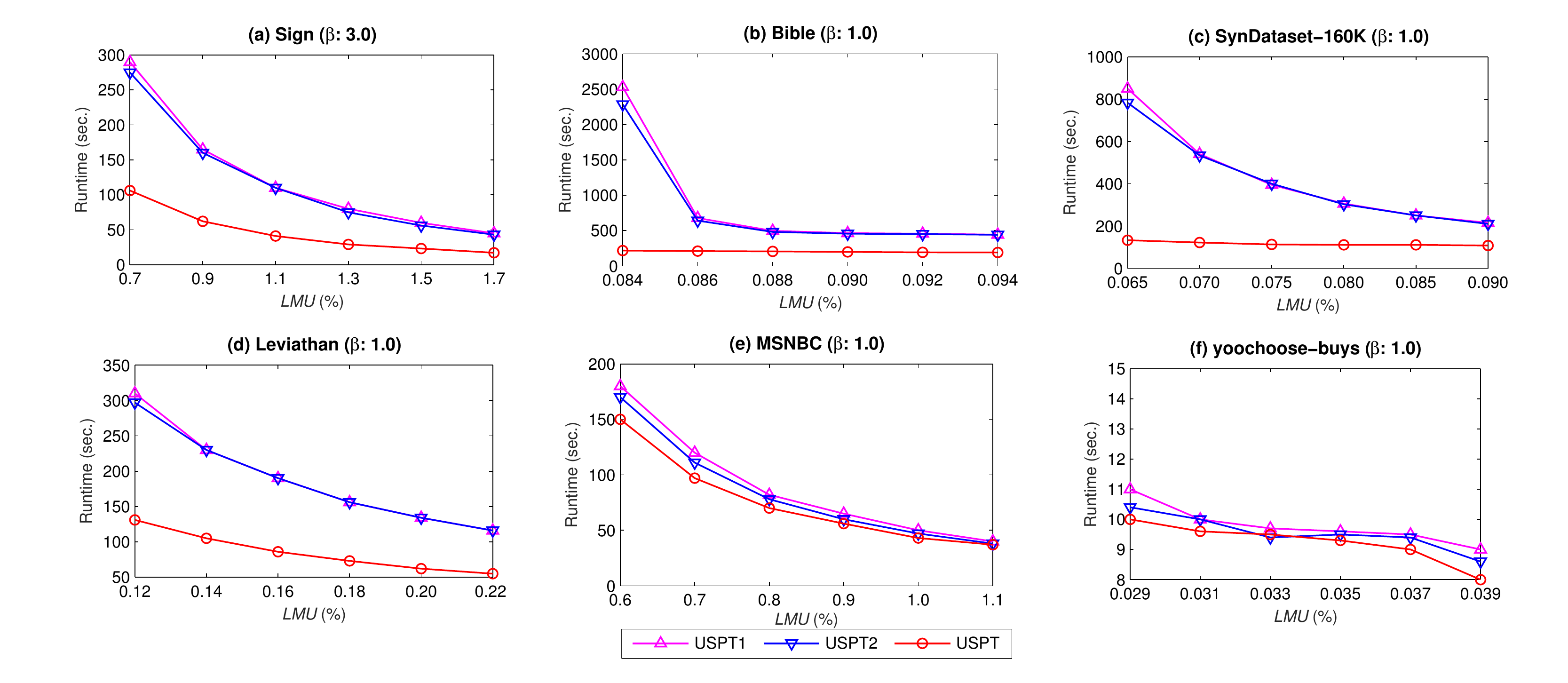}
	\caption{Runtime under various \textit{LMU} with a fixed $\beta$.}
	\label{fig:runtime1}
\end{figure}

\begin{figure}[!htbp]
	\setlength{\abovecaptionskip}{0pt}
	\setlength{\belowcaptionskip}{0pt}		
	\centering
	\includegraphics[trim=30 10 30 5,clip,scale=0.44]{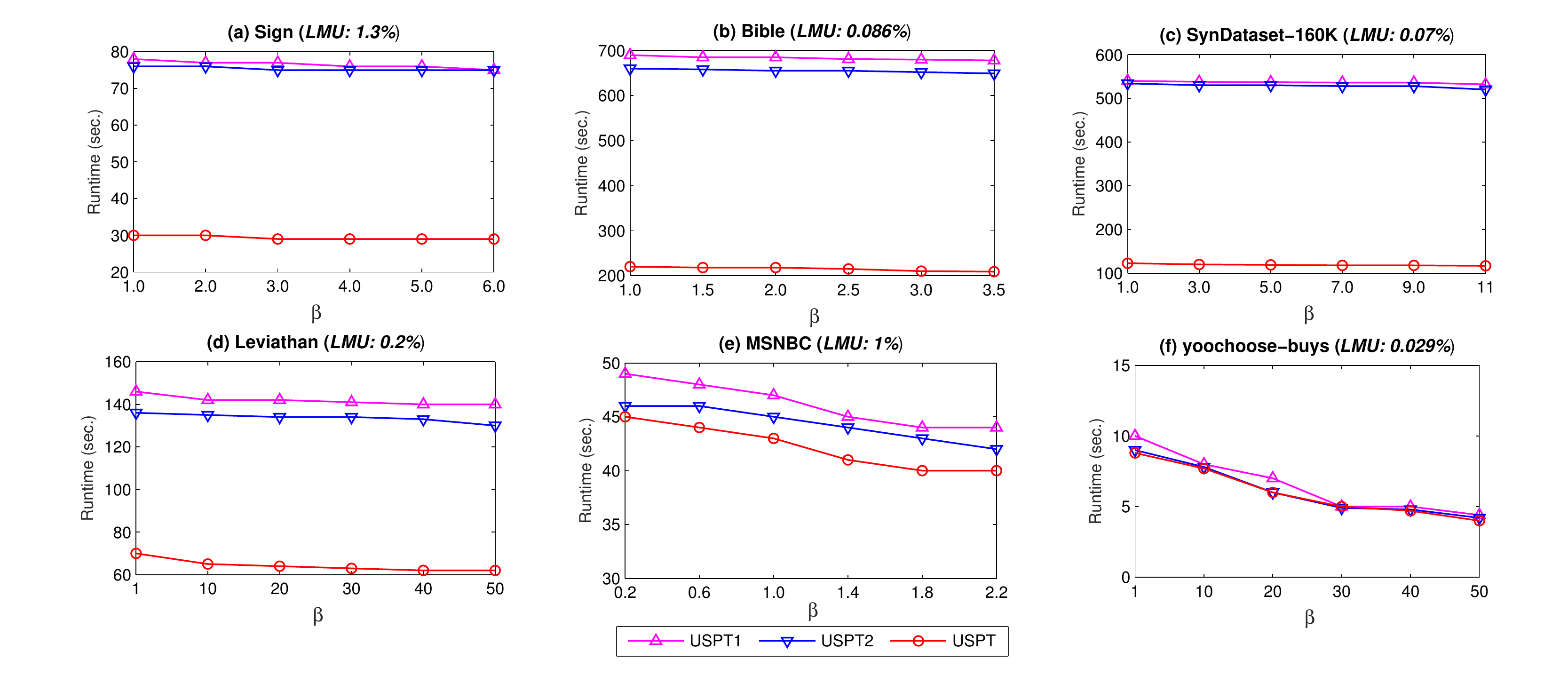}
	\caption{Runtime under various $\beta$ with a fixed \textit{LMU}.}
	\label{fig:runtime2}
\end{figure}

As shown in Figs. \ref{fig:runtime1} and \ref{fig:runtime2}, the hybrid  USPT which utilizes all pruning strategies outperformed the basic USPT1 and USPT2 on all the datasets. The runtime always shows that USPT $<$ USPT2 $<$ USPT1. Specifically, USPT was up to one or more orders of magnitude faster than the baseline algorithm without any improvements. For example, when varying \textit{LMU} from 0.065\% to 0.090\% on the SynDataset-160K dataset, as shown in Fig. \ref{fig:runtime1}(c), the runtime of USPT changed from 130 seconds to 110 seconds, while the consumed runtime of USPT1 changed from 850 seconds to 220 seconds. The same trend can also be observed in Figs. \ref{fig:runtime1}(a) to \ref{fig:runtime1}(f), and also in Figs. \ref{fig:runtime2}(a) to \ref{fig:runtime2}(f). The reason for this is that USPT utilizes a tighter upper-bound on sequence utility (e.g., \textit{PEU}) to prune the unpromising candidates earlier; it is more powerful than the sequence-weighted utilization (\textit{SWU}) that was used in the basic USPT1 algorithm. Therefore, the proposed three pruning strategies are shown to be effective and efficient to speed up mining performance.

As shown in Fig. \ref{fig:runtime1}, the runtime for all the implemented algorithms increased along with a decrease of the \textit{LMU} on all conducted datasets. When the \textit{LMU} was set lower, the runtime of USPT remained stable, but the runtime of the baseline USPT1 and USPT2 algorithms dramatically increased. This is reasonable since a large number of candidates are needed and generated with a loose upper-bounds utility, which can be seen in Figs. \ref{fig:runtime1}(a), \ref{fig:runtime1}(b), and \ref{fig:runtime1}(c).

Moreover, it is interesting to notice that the runtime of the compared approaches always decreased with the increase of $\beta$, which is shown in Figs. \ref{fig:runtime2}(a) to \ref{fig:runtime2}(f). This is because the constant factor $\beta$ increases the $ mu $ values of items. The least minimum utility thresholds of the desired sequential patterns would also be increased. Thus, less runtime was required to find fewer candidates for mining HUSPs. The developed UPST algorithm could discover the final results within a reasonable time, and the enhanced variants could achieve good performance in terms of runtime.

\subsection{Pattern Analysis}
With the same parameter settings as in Section 6.2, the size of the derived candidates and the final set of HUSPs were evaluated. The goal was to analyze the effectiveness of the proposed USPT algorithm and the effect of three pruning strategies. Note that \#\textit{Candidate} 1, \#\textit{Candidate} 2, and \#\textit{Candidate} 3 respectively showed the size of candidates generated by the USPT1, USPT2, and USP algorithms, and \# \textit{HUSPs} denoted the number of the final discovered HUSPs. The pattern results with a fixed $\beta$ and various \textit{LMUs} and with a fixed \textit{LMU} and various $\beta$ are illustrated in Tables \ref{table:candidate1} and \ref{table:candidate2}, respectively.

\begin{table}[htb]
	\fontsize{7.2pt}{9pt}\selectfont
	\centering
	\caption{Number of patterns under a fixed $\beta$ with various \textit{LMU}.}
	\label{table:candidate1}
	\begin{tabular}{ccllllll}
		\hline\hline
		\multirow{2}*{\textbf{}}&
		\multirow{2}*{\textbf{}}
		&\multicolumn{6}{c}{\textbf{Fixed $\beta$ with various \textit{LMU}}}\\
		\cline{3-8}
		&& \textit{LMU}$_1$ & \textit{LMU}$_2$ & \textit{LMU}$_3$ &  \textit{LMU}$_4$ &  \textit{LMU}$_5$  &  \textit{LMU}$_6$ \\ \hline

		&  \textbf{\#\textit{Candidate} 1} & 30,968,191 & 15,039,951	 & 8,494,306 & 	5,265,822 & 	3,490,865 & 	2,418,798  \\
		($a$)  Sign &  \textbf{\#\textit{Candidate} 2} & 30,931,811 & 	14,969,330 & 	8,433,278 & 	5,230,673 & 	3,453,759 & 	2,417,460 \\
		($\beta$: 3.0) &\textbf{\#\textit{Candidate} 3} & 4,431,028 & 	2,122,378 & 	1,179,664 & 	722,378	 & 475,204 & 	 330,575	\\
		&  \textbf{\#\textit{HUSPs}} & 84,529 & 	59,062 & 	34,479 & 	24,481 & 	17,352 & 	10,970	 \\
		\hline

		&  \textbf{\#\textit{Candidate} 1} & 1,128,248,314	 & 108,332,778 & 	14,062,212 & 	7,764,230 & 	6,997,687 & 	6,548,176 \\
		($b$)   Bible &  \textbf{\#\textit{Candidate} 2} & 999,021,192 & 	91,339,687 & 	12,677,343 & 	7,503,766 & 	6,806,316 & 	6,368,189	 \\
		($\beta$: 1.0) &\textbf{\#\textit{Candidate} 3} & 1,260,949 & 	1,042,111 & 	980,736 & 	927,481	 & 878,041 & 	832,596	\\
		&  \textbf{\#\textit{HUSPs}} & 6,001 & 	5,809	 & 5,593 & 	5,345 & 	5,108 & 	4,940 	 \\
		\hline

		&  \textbf{\#\textit{Candidate} 1} & 8,752,612 & 	5,357,804 & 	3,313,643 & 	2,122,959 & 	1,394,948 & 	948,352 \\
		($c$) Scalability-160K &  \textbf{\#\textit{Candidate} 2} & 8,736,281 & 	5,347,225 & 	3,307,685 & 	2,117,095 & 	1,389,002 & 	944,719 \\
		($\beta$: 1.0) &\textbf{\#\textit{Candidate} 3} & 252,361	 & 109,026	 & 50,308 & 	29,337 & 	21,711 & 	17,516	\\
		&  \textbf{\#\textit{HUSPs}} & 58,519 & 	17,741 & 	4,659 & 	1,220 & 	289	 & 90	 \\
		\hline

		&  \textbf{\#\textit{Candidate} 1} & 16,653,274 & 	10,697,194 & 	7,626,524 & 	5,643,813 & 	4,338,173 & 	3,421,703 \\
		($d$)  Leviathan &  \textbf{\#\textit{Candidate} 2} & 16,042,628 & 	10,266,172 & 	7,295,498 & 	5,413,709 & 	4,149,152 & 	3,274,113 \\
		($\beta$: 1.0)&\textbf{\#\textit{Candidate} 3} & 3,000,848 & 	2,025,511 & 	1,447,670 & 	1,078,496 & 	829,904	 & 654,957	\\
		&  \textbf{\#\textit{HUSPs}} & 384,681 & 	238,670	 & 157,953 & 	109,295	 & 80,857 & 	61,440 	 \\
		\hline

		&  \textbf{\#\textit{Candidate} 1} & 2,321,973 & 	1,131,628 & 	625,801 & 	385,761 & 	247,973	 & 167,433 \\
		($e$) MSNBC &  \textbf{\#\textit{Candidate} 2} & 2,321,973 & 	1,131,628 & 	625,801	 & 385,761 & 	247,973 & 	167,433	 \\
		($\beta$: 1.0)&\textbf{\#\textit{Candidate} 3} & 951,843 & 	471,188	 & 269,812 & 	171,706 & 	110,153 & 	75,437	\\
		&  \textbf{\#\textit{HUSPs}} & 136 & 	136	 & 136	 & 136	 & 122	 & 113	 \\
		\hline

		&  \textbf{\#\textit{Candidate} 1} & 268,788 & 	226,033	 & 179,922 & 	134,105	 & 88,400 & 	42,266 \\
		($f$)  yoochoose-buys &  \textbf{\#\textit{Candidate} 2} & 268,776 & 	226,018 & 	179,905 & 	134,093	 & 88,395 & 	42,255	 \\
		($\beta$: 1.0) &\textbf{\#\textit{Candidate} 3} & 268,408	 & 225,692 & 	179,611	 & 133,832 & 	88,152 & 	42,048	\\
		&  \textbf{\#\textit{HUSPs}} & 258,886 & 	214,309 & 	168,215	 & 122,706 & 	76,789 & 	31,142	 \\
		
		\hline\hline
	\end{tabular}
\end{table}

\begin{table}[htb]
	\fontsize{7.2pt}{9pt}\selectfont
	\centering
	\caption{Number of patterns under a fixed \textit{LMU} with various $\beta$.}
	\label{table:candidate2}
	\begin{tabular}{ccllllll}
		\hline\hline
		\multirow{2}*{\textbf{}}&
		\multirow{2}*{\textbf{}}
		&\multicolumn{5}{c}{\textbf{Fixed \textit{LMU} with various $\beta$}}\\
		\cline{3-8}
		&&$ \beta_1 $ & $ \beta_2 $ & $ \beta_3 $ & $ \beta_4 $ &  $ \beta_5 $  &  $ \beta_6 $ \\ \hline

		&  \textbf{\#\textit{Candidate} 1} & 5,265,822 & 	5,265,822 & 	5,265,822 & 	5,265,822 & 	5,265,822 & 	5,265,795 \\
		($a$)  Sign &  \textbf{\#\textit{Candidate} 2} & 5,230,673 & 	5,230,673 & 	5,230,673 & 	5,230,673 & 	5,230,673 & 	5,230,646	 \\
		(\textit{LMU}: 1.3\%) &\textbf{\#\textit{Candidate} 3} & 722,378 & 	722,378	 & 722,378 & 	722,378	 & 722,378	 & 722,376	\\
		&  \textbf{\#\textit{HUSPs}} & 55,883 & 	45,733 & 	24,481 & 	9,561 & 	4,733 & 	1,517	 \\
		\hline

		&  \textbf{\#\textit{Candidate} 1} & 108,332,778 & 	108,332,778 & 	108,332,778 & 	108,329,765	 & 108,329,702	 & 108,329,702 \\
		($b$)   Bible &  \textbf{\#\textit{Candidate} 2} & 91,339,687	 & 91,339,687 & 	91,339,687 & 	91,336,966 & 	91,336,906 & 	91,336,906	 \\
		(\textit{LMU}: 0.086\%) &\textbf{\#\textit{Candidate} 3} & 1,042,111 & 	1,042,111 & 	1,042,111 & 	1,042,111 & 	1,042,111 & 	1,042,111	\\
		&  \textbf{\#\textit{HUSPs}} & 5,809 & 	1,413 & 	503	 & 283 & 	167 & 	88	 \\
		\hline

		&  \textbf{\#\textit{Candidate} 1} & 5,357,804	 & 5,357,804 & 	5,357,804 & 	5,357,804 & 	5,357,804 & 	5,357,804 \\
		($c$)    Scalability-160K &  \textbf{\#\textit{Candidate} 2} & 5,347,225	 & 5,347,225 & 	5,347,225 & 	5,347,225 & 	5,347,225 & 	5,347,225	 \\
		(\textit{LMU}: 0.07\%) &\textbf{\#\textit{Candidate} 3} & 109,026	 & 109,026	 & 109,026 & 	109,026 & 	109,026 & 	109,026 	\\
		&  \textbf{\#\textit{HUSPs}} & 17,741 & 	17,513 & 	16,737 & 	8,982 & 	352 & 	261	 \\
		\hline

		&  \textbf{\#\textit{Candidate} 1} & 4,338,173	 & 4,338,112	 & 4,337,591 & 	4,332,859 & 	4,290,103 & 	4,093,436 \\
		($d$)   Leviathan &  \textbf{\#\textit{Candidate} 2} & 4,149,152 & 	4,149,096 & 	4,148,651 & 	4,144,219 & 	4,104,213 & 	3,919,454	 \\
		(\textit{LMU}: 0.2\%) &\textbf{\#\textit{Candidate} 3} & 829,904 & 	829,873 & 	829,749 & 	829,269	 & 826,212 & 	811,885	\\
		&  \textbf{\#\textit{HUSPs}} & 80,857 & 	78,006 & 	71,399 & 	19,048 & 	8,329 & 	2,052	 \\
		\hline

		&  \textbf{\#\textit{Candidate} 1} & 248,017 & 	248,017 & 	247,973 & 	234,225 & 	231,471 & 	230,921  \\
		($e$)   MSNBC &  \textbf{\#\textit{Candidate} 2} & 248,017 & 	248,017	 & 247,973 & 	234,225 & 	231,471 & 	230,921	 \\
		(\textit{LMU}: 1\%)  &\textbf{\#\textit{Candidate} 3} & 110,156	 & 110,156 & 	110,153 & 	102,942 & 	101,690 & 	101,582	\\
		&  \textbf{\#\textit{HUSPs}} & 7,539 & 	1,005	 & 122 & 	32 & 	7 & 	2	 \\
		\hline

		&  \textbf{\#\textit{Candidate} 1} & 268,788 & 	266,582	 & 266,011 & 	253,744 & 	102,749 & 	404\\
		($f$)   yoochoose-buys &  \textbf{\#\textit{Candidate} 2} & 268,776 & 	266,573 & 	266,005	 & 253,743 & 	102,747 & 	402	 \\
		(\textit{LMU}: 0.029\%) &\textbf{\#\textit{Candidate} 3} & 268,408 & 	266,372 & 	265,919	 & 253,680 & 	102,709 & 	364	\\
		&  \textbf{\#\textit{HUSPs}} & 258,886 & 	255,471 & 	255,471	 & 230,207 & 	87,290 & 	0 \\

		\hline\hline
	\end{tabular}
\end{table}

As shown in Tables \ref{table:candidate1} and \ref{table:candidate2}, the size of the candidates generated by USPT is much fewer that of USPT1 and USPT2. For instance, in the Leviathan dataset, shown in Table \ref{table:candidate1}(d), when $\beta$ was set as 1.0 and \textit{LMU} was set as 0.20\%, the results respectively were \#\textit{Candidate} 1: 3,421,703, \#\textit{Candidate} 2: 3,274,113, and \#\textit{Candidate} 3: 654,957. Among these, only \#\textit{Candidate} 3 was close to the final results of the HUSPs since \#\textit{HUSPs}: 61,440. Therefore, the developed strategies could be used to greatly reduce the size of candidates in the search space for mining the actual HUSPs. The runtime for mining the HUSPs could also be improved. More specifically, the results of the size of candidates imply that the more powerful filtering ability of the enhanced USPT algorithm pruned the unpromising patterns. The runtime in Figs. \ref{fig:runtime1} and \ref{fig:runtime2} show that the developed strategies could improve mining performance because the size of the candidates was reduced in the search space. To obtain the tight upper-bound values of the candidates, USPT1 and USPT2 need to recalculate the remaining utilities in the utility-array structure, which requires the additional computations.

From the above results, it can be observed that with the decrease of the \textit{LMU} and $ \beta $, the number of candidates increases for the USPT, USPT1, and USPT2. However, with the effect of the pruning strategies, USPT and USPT2 have better performances than the baseline USPT1. Therefore, USPT produced much fewer candidates than USPT1 when a lower \textit{LMU} or $ \beta $ was set. It can also be observed that the size of the candidates sharply decreases in Tables \ref{table:candidate1}(b) and \ref{table:candidate1}(e), which is very interesting and reassuring.

\subsection{Memory Consumption}
The maximal memory consumption of the compared approaches was  further evaluated in this subsection. Note that the maximal memory consumption was checked using Java API. Results under various \textit{LMUs} with a fixed $\beta$ of memory consumption are shown in Fig. \ref{fig:memory1}. Moreover, results of memory consumption under various $\beta$ with a fixed \textit{LMU} are shown in Fig. \ref{fig:memory2}.

\begin{figure}[!htbp]
	\setlength{\abovecaptionskip}{0pt}
	\setlength{\belowcaptionskip}{0pt}	
	\centering
	\includegraphics[trim=30 10 30 5,clip,scale=0.45]{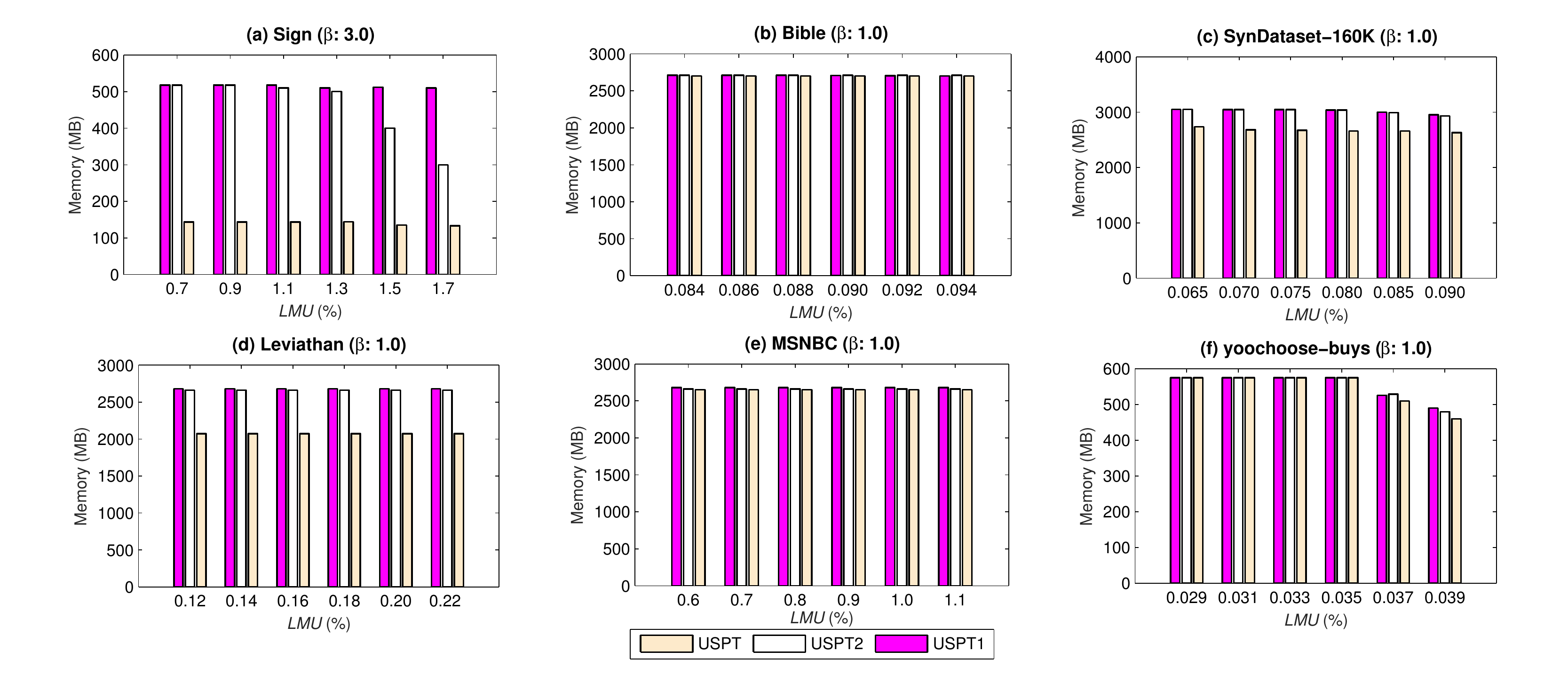}
	\caption{Memory consumption under various \textit{LMU} with a fixed $\beta$.}
	\label{fig:memory1}
\end{figure}

\begin{figure}[!htbp]
	\setlength{\abovecaptionskip}{0pt}
	\setlength{\belowcaptionskip}{0pt}		
	\centering
	\includegraphics[trim=30 10 35 5,clip,scale=0.45]{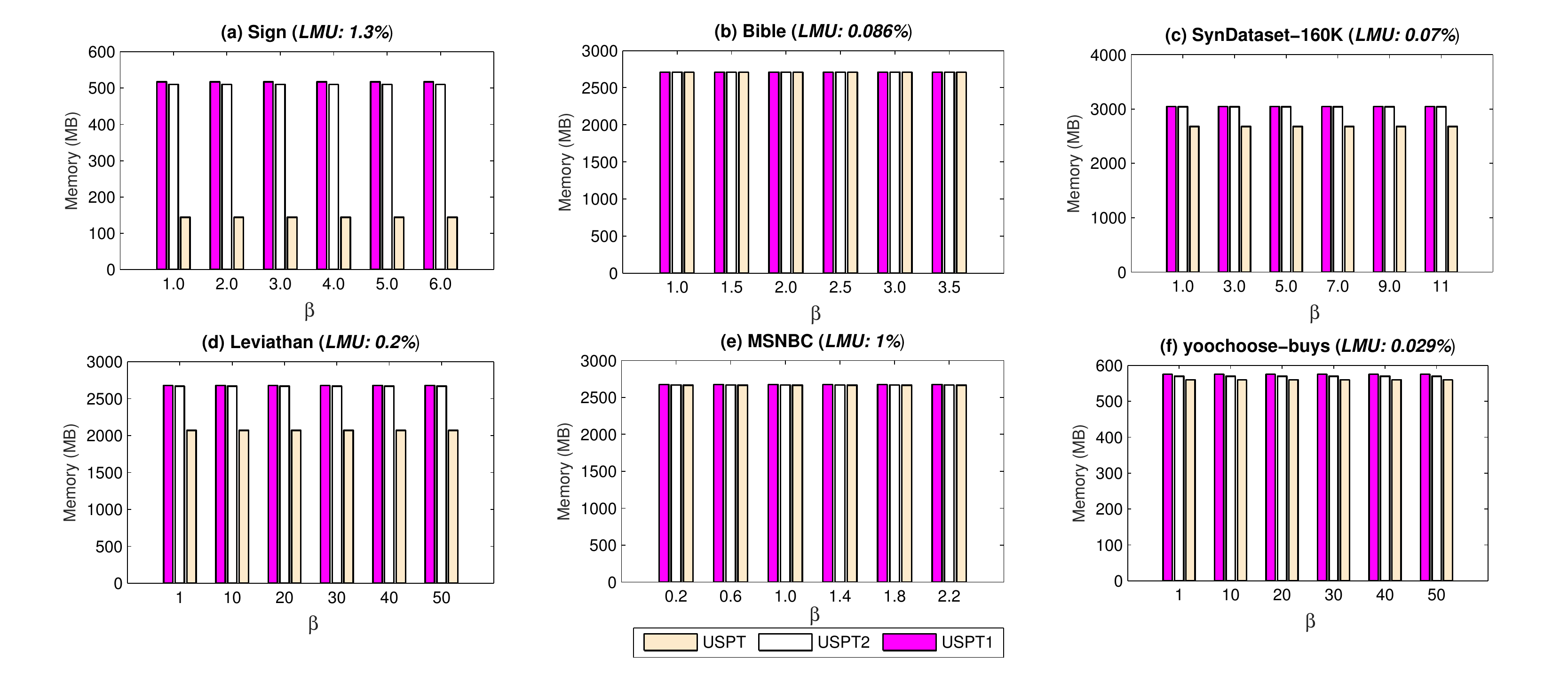}
	\caption{Memory consumption under various $\beta$ with a fixed \textit{LMU}.}
	\label{fig:memory2}
\end{figure}

Obviously, three variants of the proposed USPT algorithm consume the stable memory to discover HUSPs with individualized minimum utility thresholds. Among them, USPT1 and USPT2 always consume the similar maximal memory in different parameter settings on all datasets. And the hybrid  USPT variant always consumes the least memory on all datasets, and sometimes the gap of memory consumption between USPT and other baselines is large, as shown in Sign (Fig. \ref{fig:memory1}(a) and Fig. \ref{fig:memory2}(a)) and Leviathan (Fig. \ref{fig:memory1}(d) and Fig. \ref{fig:memory2}(d)). According the above experimental results of runtime, the number of generated candidates and memory consumption, we can summarize the conclusion as follows: the designed USPT framework can successfully discover the set of high-utility sequential patterns with individualized thresholds, and the pruning strategies can further be utilized to speedup mining efficiency, in terms of execution time, powerful of pruning unpromising candidates, and the memory consumption.

\subsection{USPT Framework vs. Basic HUSPM Model}

As mentioned before, a traditional HUSPM algorithm could be a special case of the USPT problem statement whereas the minimum individualized thresholds equal to the unified threshold in HUSPM. In other words, when $\beta$ in  the function $mu(i_{j}) = max\{\beta\times u(i_{j}), LMU\}$ is set to 0, then the \textit{mu} value of each item/sequence is equal to \textit{LMU}. Thus, the basic HUSPM model has become a special case of the USPT problem statement when $\beta$ = 0. In this subsection, a performance comparison of USPT and a traditional HUSPM algorithm was provided. We selected HUS-Span \cite{wang2016efficiently} as the compared HUSPM algorithm. With the same parameter settings as in Section 6.2, detailed results are presented in Table \ref{table:withHUSPM}(a) to \ref{table:withHUSPM}(f), respectively. Note that $\beta$ is set to 0 in USPT, and the \textit{LMU} equals the unified minimum utility threshold in HUS-Span.

\begin{table}[htb]
	\fontsize{7.2pt}{8pt}\selectfont
	\centering
	\caption{Number of patterns under a fixed $\beta$ with various \textit{LMU}.}
	\label{table:withHUSPM}
	\begin{tabular}{ccllllll}
		\hline\hline
		\multirow{2}*{\textbf{}}&
		\multirow{2}*{\textbf{}}
		&\multicolumn{6}{c}{\textbf{Fixed $\beta$ with various \textit{LMU}}}\\
		\cline{3-8}
		&& \textit{LMU}$_1$ & \textit{LMU}$_2$ & \textit{LMU}$_3$ &  \textit{LMU}$_4$ &  \textit{LMU}$_5$  &  \textit{LMU}$_6$ \\ \hline

		&  \textbf{\textit{Time*}} & 643	 & 416	 & 278 & 	206 & 	157	 & 119  \\

		&  \textbf{\textit{Memory*}} & 2,063	 & 2,067 & 	2,067 & 	2,064 & 	2,062 & 	2,064 \\

 &\textbf{\#\textit{Candidate*}} & 30,971,048 & 	15,039,955 & 	8,494,312 & 	5,265,825 & 	3,490,869 & 	2,418,799	\\
		($a$) Sign 			&  \textbf{\#\textit{HUSPs*}} & 621,291 & 	245,689	 & 112,001 & 	56,395 & 	30,440 & 	17,274	 \\
		\cline{2-8}
	($\beta$: 3.0)			&  \textbf{\textit{Time}} & 115 & 	69 & 	46 & 	33 & 	23 & 	19 \\
	  &  \textbf{\textit{Memory}} & 144 & 	144	 & 144 & 	133 & 	133 & 	133 \\
		 &\textbf{\#\textit{Candidate}} & 4,431,158 & 	2,122,378 & 	1,179,664 & 	722,378 & 	475,204	 & 330,575	\\		
		&  \textbf{\#\textit{HUSPs}} & 621,291 & 	245,689	 & 112,001 & 	56,395 & 	30,440 & 	17,274	 \\
		\hline

		&  \textbf{\textit{Time*}} & - & 	1,680 & 	1,172 & 	1,088 & 	1,048 & 	1,030  \\

&  \textbf{\textit{Memory*}} & - &	3,064 &	2,966 &	3,055 &	2,987 &	2,900  \\
 &\textbf{\#\textit{Candidate*}} & - &	108,260,940	 &14,012,881 &	7,659,519 &	6,916,037 &	6,471,482	\\
($b$)  Bible  &  \textbf{\#\textit{HUSPs*}} & 324,599 &	305,126 &	287,227 &	270,850 &	255,662 &	241,882	 \\
\cline{2-8}

($\beta$: 1.0) &  \textbf{\textit{Time}} & 213 & 	210	 & 198 &	192 &	189 &	208  \\
&  \textbf{\textit{Memory}} & 2,708 &	2,708 &	2,707 &	2,707 &	2,705 &	2,705  \\

&\textbf{\#\textit{Candidate}} & 1,260,949 &	1,042,111 &	980,736 &	927,481 &	878,041 &	832,596	\\		
&  \textbf{\#\textit{HUSPs}} & 324,599 &	305,126 &	287,227 &	270,850 &	255,662 &	241,882	 \\
\hline

&  \textbf{\textit{Time*}} & 2,824 &	2,022 &	1,460 &	1,152 &	950 &	820  \\

&  \textbf{\textit{Memory*}} & 2,021 &	1,830 &	1,734 &	1,716 &	1,700 &	1,512 \\
&\textbf{\#\textit{Candidate*}} & 8,752,654 &	5,357,847 &	3,313,681 &	2,123,007 &	1,394,913 &	948,397	\\
($c$)  SynDataset-160K &  \textbf{\#\textit{HUSPs*}} & 58,710 &	17,903 &	4,794 &	1,344 &	394	172	 \\
\cline{2-8}

($\beta$: 1.0)  &  \textbf{\textit{Time}} & 130 &	119 &	114 &	112 &	109 &	108 \\
 &  \textbf{\textit{Memory}} & 2,734 &	2,679 &	2,667 &	2,665 &	2,665 &	2,663 \\

&\textbf{\#\textit{Candidate}} & 252,361 &	109,026 &	50,308 &	29,337 &	21,711 &	17,516 	\\		
&  \textbf{\#\textit{HUSPs}} & 58,710 &	17,903 &	4,794 &	1,344 &	394	172	 \\
\hline

&  \textbf{\textit{Time*}} & 647 &	474	 &435 &	351 & 289 &	263  \\

&  \textbf{\textit{Memory*}} & 2,720	 &2,671	 & 2,671 &	2,670 &	2,669 &	2,114  \\

&\textbf{\#\textit{Candidate*}} & 16,498,302 &	10,628,212 &	7,571,368 &	5,601,623 &	4,307,997 &	3,401,042 	\\
($d$)  Leviathan  &  \textbf{\#\textit{HUSPs*}} & 1,034,006 &	672,862 &	463,814 &	333,084 &	249,942	 & 192,424	 \\
\cline{2-8}

($\beta$: 1.0) &  \textbf{\textit{Time}} & 132 &	104 &	85 &	72 &	61 &	53 \\
&  \textbf{\textit{Memory}} & 2,677 &	2,677 &	2,071 &	2,070 &	2,071 &	2,072  \\
&\textbf{\#\textit{Candidate}} & 3,000,848 &	2,025,511 &	1,447,670 &	1,078,496 &	829,904 &	654,957 	\\		
&  \textbf{\#\textit{HUSPs}} & 1,034,006 &	672,862 &	463,814 &	333,084 &	249,942	 & 192,424	 \\
\hline

&  \textbf{\textit{Time*}} & 1,015 &	498 &	294 &	201 &	144 &	111  \\

&  \textbf{\textit{Memory*}} & 2,722 &	2,730 &	2,760 &	2,755 &	2,748 &	2,763 \\

&\textbf{\#\textit{Candidate*}} & 3,041,187 &	1,347,748 &	686,318 &	392,296 &	247,873 &	167,378	\\
($e$)  MSNBC  &  \textbf{\#\textit{HUSPs*}} & 113,728 &	51,007 &	30,313 &	20,058 &	14,275 &	10,742	 \\
\cline{2-8}

($\beta$: 1.0)  &  \textbf{\textit{Time}} & 210 &	124 &	77 &	62 &	43 &	37 \\
&  \textbf{\textit{Memory}} & 2,674 &	2,673 &	2,674 &	2,676 &	2,677 &	2,121  \\
&\textbf{\#\textit{Candidate}} & 1,364,376 &	590,177 &	297,371 &	175,122 &	110,156 &	75,439	\\		
&  \textbf{\#\textit{HUSPs}} & 113,728 &	51,007 &	30,313 &	20,058 &	14,275 &	10,742	 \\
\hline

&  \textbf{\textit{Time*}} & 2.8 &	2.7	 &2.7 &	2.6 &	2.6 &	2.4  \\

&  \textbf{\textit{Memory*}} & 1,165 &	1,083 &	1,083 &	892 &	612 &	570 \\

 &\textbf{\#\textit{Candidate*}} & 269,036 &	226,274 &	180,138 &	134,329 &	88,623 &	42,476	\\
($f$)  yoochoose-buys &  \textbf{\#\textit{HUSPs*}} & 259,765 &	215,098 &	168,945 & 123,368 &	77,404 &	31,720	 \\
\cline{2-8}

($\beta$: 1.0) &  \textbf{\textit{Time}} & 10 &	10 &	9.7 &	10 &	9.5 &	9 \\
 &  \textbf{\textit{Memory}} & 575 &	575 &	575 &	575 &	525 &	484 \\
&\textbf{\#\textit{Candidate}} & 268,498 &	225,760 &	179,670 &	133,885 &	88,202 &	42,094	\\		
&  \textbf{\#\textit{HUSPs}} & 259,765 &	215,098 &	168,945 & 123,368 &	77,404 &	31,720	 \\
\hline

		\hline\hline
	\end{tabular}
\end{table}

Note that \textit{Time*}, \textit{Memory*}, \#\textit{Candidate*}, and \#\textit{HUSPs*} are the results of HUS-Span, and \textit{Time}, \textit{Memory}, \#\textit{Candidate}, and \#\textit{HUSPs} are related to USPT. Note that the scale of the runtime and memory are second and MB, respectively. From Table \ref{table:withHUSPM}(a) to Table \ref{table:withHUSPM}(f), their final set of discovered HUSPs are the same, while USPT always significantly outperforms than HUS-Span under various \textit{LMU}. For example, USPT consumes less execution time and memory than that of  HUS-Span except for yoochoose-buys. In particular, the number of candidates generated by USPT is more close to the final set of HUSPs. For instance, in Sign, when \textit{LMU} is set to 1.3\%, the results are: \#\textit{Candidate*}: 5,265,825, \#\textit{Candidate}: 722,378, and \#\textit{HUSPs}: 56,395. Obviously,   \#\textit{Candidate} is quite less than  \#\textit{Candidate*}, and close to \#\textit{HUSPs}. Notice that there are no results in Tables \ref{table:withHUSPM}(b) for the HUS-Span algorithm when \textit{LMU}: 0.084\%. This is because HUS-Span was terminated when the mining task required more than 10,000 seconds.  We can conclude that the proposed USPT approach is a general utility mining framework and has an acceptable effectiveness, since it can not only address the HUSPM problem with individualized thresholds, but also deal with the basic HUSPM problem with the unified threshold. In addition to the effectiveness, USPT has better efficiency than the state-of-the-art HUSPM algorithm.

\subsection{Scalability Test}
The scalability of the compared approaches was evaluated on a synthetic dataset named C8S6T4I3D$ | $X$ | $K. The results in terms of runtime and the number of candidates and HUSPs under the fixed \textit{LMU} and $\beta$ are shown in Fig. \ref{fig:scalability}. The dataset size of C8S6T4I3D$ | $X$ | $K is varied from 40K to 400K.

\begin{figure}[!htbp]
	\setlength{\abovecaptionskip}{0pt}
	\setlength{\belowcaptionskip}{0pt}	
	\centering
	\includegraphics[trim=25 200 30 0,clip,scale=0.62]{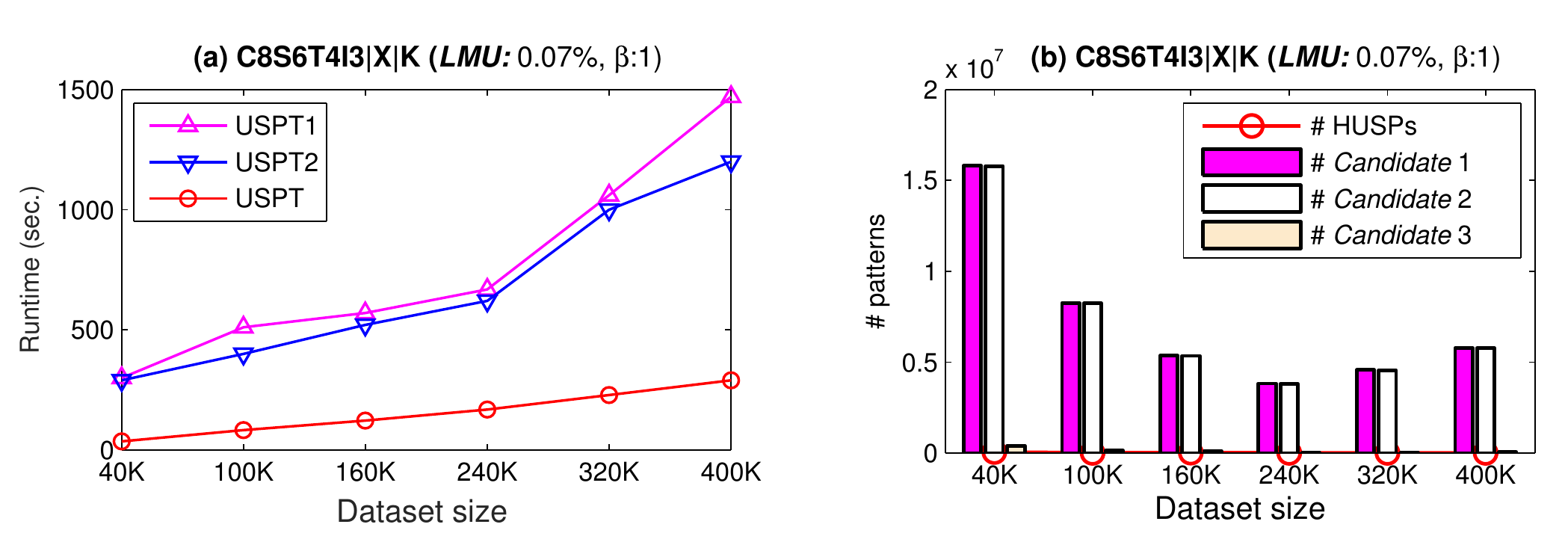}
	\caption{Scalability of the compared approaches.}
	\label{fig:scalability}
\end{figure}

As shown in Fig. \ref{fig:scalability}, the USPT algorithm has obtained better results in terms of scalability than that of the USPT1 and USPT2 algorithms. With the increasing dataset size, the runtime and the number of candidates of USPT1, USPT2, and USPT increase as well. However, the runtime gap between USPT and USPT1 has become larger along with the increased dataset size. This is reasonable since when the size of dataset has become larger or more denser, more candidates in transactions are required to be determined for mining the HUSPs. This process leads to an increasing in the requirements of runtime consumption. Without the designed strategies to efficiently prune the redundant candidates in the search space, the baseline algorithm sometimes requires a long time to perform the mining process. As shown in Fig. \ref{fig:scalability}(b), the size of the candidates remains stable along with the size of the dataset. The reason is that when the minimum utility value increases along with the dataset size, fewer candidates satisfy the condition for being HUSPs. If the data distribution in datasets was uniform, the number of candidates would increase or decrease steadily. However, the results in Fig. \ref{fig:scalability}(b) shows that the datasets had no uniformity in data distribution. A computation to examine the candidates was still required, and the runtime increases along with the increasing size of the dataset.

\section{Conclusion and Future Works}
\label{sec:7}
To solve utility-oriented sequential pattern mining, we presented a new USPT framework for efficiently discovering the set of HUSPs across multi-sequences by regarding individualized minimum utility thresholds, which are more informative and useful for decision makers, managers, and experts. To the best of our knowledge, this is the first study for solving the USPT problem. By utilizing the developed properties of upper-bounds, LS-tree and utility-array structures, the USPT framework was first proposed to extract the complete set of HUSPs. By utilizing the properties of various upper-bounds, three strategies were designed to prune the search space and improve the mining performance. To show the effectiveness and efficiency of the designed USPT algorithm, several experiments were conducted.

Since this is the first study on utility mining across multi-sequences with individualized thresholds, many extensions could be explored and studied in the future, such as designing more efficient algorithms, applying the proposed USPT algorithm in big data \cite{4gan2017data}, or extending the USPT model to other interesting applications (e.g., discovering rare pattern \cite{koh2016unsupervised}, parallel mining  \cite{3gan2018survey}, privacy preserving utility mining \cite{gan2018privacy}, and recommender system \cite{yin2014lcars}). Moreover, mining HUSPs with individualized thresholds over uncertain data \cite{lin2016efficient} would be a non-trivial and challenging task.

\section{Acknowledgment}

We sincerely thank the editors and anonymous reviewers whose valuable comments helped us to improve this paper significantly. This research was partially supported by the Shenzhen Technical Project under KQJSCX 20170726103424709 and JCYJ 20170307151733005, and NSF under grants III-1526499, III-1763325, III-1909323, SaTC-1930941, and CNS-1626432. Specifically, Wensheng Gan was supported by the China Scholarship Council Program, during his study at University of Illinois at Chicago, IL, USA.


\bibliographystyle{ACM-Reference-Format}
\bibliography{main}


\begin{thebibliography}{60}


\ifx \showCODEN    \undefined \def \showCODEN     #1{\unskip}     \fi
\ifx \showDOI      \undefined \def \showDOI       #1{#1}\fi
\ifx \showISBNx    \undefined \def \showISBNx     #1{\unskip}     \fi
\ifx \showISBNxiii \undefined \def \showISBNxiii  #1{\unskip}     \fi
\ifx \showISSN     \undefined \def \showISSN      #1{\unskip}     \fi
\ifx \showLCCN     \undefined \def \showLCCN      #1{\unskip}     \fi
\ifx \shownote     \undefined \def \shownote      #1{#1}          \fi
\ifx \showarticletitle \undefined \def \showarticletitle #1{#1}   \fi
\ifx \showURL      \undefined \def \showURL       {\relax}        \fi
\providecommand\bibfield[2]{#2}
\providecommand\bibinfo[2]{#2}
\providecommand\natexlab[1]{#1}
\providecommand\showeprint[2][]{arXiv:#2}

\bibitem[\protect\citeauthoryear{Agrawal, Imielinski, and Swami}{Agrawal
  et~al\mbox{.}}{1993}]%
        {agrawal1993database}
\bibfield{author}{\bibinfo{person}{Rakesh Agrawal}, \bibinfo{person}{Tomasz
  Imielinski}, {and} \bibinfo{person}{Arun Swami}.}
  \bibinfo{year}{1993}\natexlab{}.
\newblock \showarticletitle{Database mining: A performance perspective}.
\newblock \bibinfo{journal}{\emph{IEEE Transactions on Knowledge and Data
  Engineering}} \bibinfo{volume}{5}, \bibinfo{number}{6}
  (\bibinfo{year}{1993}), \bibinfo{pages}{914--925}.
\newblock


\bibitem[\protect\citeauthoryear{Agrawal and Srikant}{Agrawal and
  Srikant}{1994}]%
        {agrawal1994dataset}
\bibfield{author}{\bibinfo{person}{Rakesh Agrawal} {and}
  \bibinfo{person}{Ramakrishnan Srikant}.} \bibinfo{year}{1994}\natexlab{}.
\newblock \bibinfo{title}{Quest synthetic data generator}.
\newblock
  \bibinfo{howpublished}{http://www.Almaden.ibm.com/cs/quest/syndata.html}.
  (\bibinfo{year}{1994}).
\newblock


\bibitem[\protect\citeauthoryear{Agrawal and Srikant}{Agrawal and
  Srikant}{1995}]%
        {agrawal1995mining}
\bibfield{author}{\bibinfo{person}{Rakesh Agrawal} {and}
  \bibinfo{person}{Ramakrishnan Srikant}.} \bibinfo{year}{1995}\natexlab{}.
\newblock \showarticletitle{Mining sequential patterns}. In
  \bibinfo{booktitle}{\emph{Proceedings of the International Conference on Data
  Engineering}}. IEEE, \bibinfo{pages}{3--14}.
\newblock


\bibitem[\protect\citeauthoryear{Agrawal, Srikant, et~al\mbox{.}}{Agrawal
  et~al\mbox{.}}{1994}]%
        {agrawal1994fast}
\bibfield{author}{\bibinfo{person}{Rakesh Agrawal},
  \bibinfo{person}{Ramakrishnan Srikant}, {et~al\mbox{.}}}
  \bibinfo{year}{1994}\natexlab{}.
\newblock \showarticletitle{Fast algorithms for mining association rules}. In
  \bibinfo{booktitle}{\emph{Proceedings of the 20th International Conference on
  Very Large Data Bases}}, Vol.~\bibinfo{volume}{1215}.
  \bibinfo{pages}{487--499}.
\newblock


\bibitem[\protect\citeauthoryear{Ahmed, Tanbeer, and Jeong}{Ahmed
  et~al\mbox{.}}{2010a}]%
        {ahmed2010mining}
\bibfield{author}{\bibinfo{person}{Chowdhury-Farhan Ahmed},
  \bibinfo{person}{Syed-Khairuzzaman Tanbeer}, {and}
  \bibinfo{person}{Byeong-Soo Jeong}.} \bibinfo{year}{2010}\natexlab{a}.
\newblock \showarticletitle{Mining high utility web access sequences in dynamic
  web log data}. In \bibinfo{booktitle}{\emph{11th ACIS International
  Conference on Software Engineering, Artificial Intelligence, Networking and
  Parallel/Distributed Computing}}. IEEE, \bibinfo{pages}{76--81}.
\newblock


\bibitem[\protect\citeauthoryear{Ahmed, Tanbeer, and Jeong}{Ahmed
  et~al\mbox{.}}{2010b}]%
        {ahmed2010novel}
\bibfield{author}{\bibinfo{person}{Chowdhury-Farhan Ahmed},
  \bibinfo{person}{Syed-Khairuzzaman Tanbeer}, {and}
  \bibinfo{person}{Byeong-Soo Jeong}.} \bibinfo{year}{2010}\natexlab{b}.
\newblock \showarticletitle{A novel approach for mining high-utility sequential
  patterns in sequence databases}.
\newblock \bibinfo{journal}{\emph{ETRI Journal}} \bibinfo{volume}{32},
  \bibinfo{number}{5} (\bibinfo{year}{2010}), \bibinfo{pages}{676--686}.
\newblock


\bibitem[\protect\citeauthoryear{Ahmed, Tanbeer, Jeong, and Lee}{Ahmed
  et~al\mbox{.}}{2009}]%
        {ahmed2009efficient}
\bibfield{author}{\bibinfo{person}{Chowdhury-Farhan Ahmed},
  \bibinfo{person}{Syed-Khairuzzaman Tanbeer}, \bibinfo{person}{Byeong-Soo
  Jeong}, {and} \bibinfo{person}{Young-Koo Lee}.}
  \bibinfo{year}{2009}\natexlab{}.
\newblock \showarticletitle{Efficient tree structures for high utility pattern
  mining in incremental databases}.
\newblock \bibinfo{journal}{\emph{IEEE Transactions on Knowledge and Data
  Engineering}} \bibinfo{volume}{21}, \bibinfo{number}{12}
  (\bibinfo{year}{2009}), \bibinfo{pages}{1708--1721}.
\newblock


\bibitem[\protect\citeauthoryear{Alkan and Karagoz}{Alkan and Karagoz}{2015}]%
        {alkan2015crom}
\bibfield{author}{\bibinfo{person}{Oznur~Kirmemis Alkan} {and}
  \bibinfo{person}{Pinar Karagoz}.} \bibinfo{year}{2015}\natexlab{}.
\newblock \showarticletitle{{CR}o{M} and {H}usp{E}xt: Improving efficiency of
  high utility sequential pattern extraction}.
\newblock \bibinfo{journal}{\emph{IEEE Transactions on Knowledge and Data
  Engineering}} \bibinfo{volume}{27}, \bibinfo{number}{10}
  (\bibinfo{year}{2015}), \bibinfo{pages}{2645--2657}.
\newblock


\bibitem[\protect\citeauthoryear{Ayres, Flannick, Gehrke, and Yiu}{Ayres
  et~al\mbox{.}}{2002}]%
        {ayres2002sequential}
\bibfield{author}{\bibinfo{person}{Jay Ayres}, \bibinfo{person}{Jason
  Flannick}, \bibinfo{person}{Johannes Gehrke}, {and} \bibinfo{person}{Tomi
  Yiu}.} \bibinfo{year}{2002}\natexlab{}.
\newblock \showarticletitle{Sequential pattern mining using a bitmap
  representation}. In \bibinfo{booktitle}{\emph{Proceedings of the 8th ACM
  SIGKDD International Conference on Knowledge Discovery and Data Mining}}.
  ACM, \bibinfo{pages}{429--435}.
\newblock


\bibitem[\protect\citeauthoryear{Chan, Yang, and Shen}{Chan
  et~al\mbox{.}}{2003}]%
        {chan2003mining}
\bibfield{author}{\bibinfo{person}{Raymond Chan}, \bibinfo{person}{Qiang Yang},
  {and} \bibinfo{person}{Yi~Dong Shen}.} \bibinfo{year}{2003}\natexlab{}.
\newblock \showarticletitle{Mining high utility itemsets}. In
  \bibinfo{booktitle}{\emph{Proceedings of the 3rd IEEE International
  Conference on Data Mining}}. IEEE, \bibinfo{pages}{19--26}.
\newblock


\bibitem[\protect\citeauthoryear{Chen, Han, and Yu}{Chen et~al\mbox{.}}{1996}]%
        {chen1996data}
\bibfield{author}{\bibinfo{person}{Ming-Syan Chen}, \bibinfo{person}{Jiawei
  Han}, {and} \bibinfo{person}{Philip~S. Yu}.} \bibinfo{year}{1996}\natexlab{}.
\newblock \showarticletitle{Data mining: an overview from a database
  perspective}.
\newblock \bibinfo{journal}{\emph{IEEE Transactions on Knowledge and Data
  Engineering}} \bibinfo{volume}{8}, \bibinfo{number}{6}
  (\bibinfo{year}{1996}), \bibinfo{pages}{866--883}.
\newblock


\bibitem[\protect\citeauthoryear{Fournier-Viger, Lin, Kiran, and
  Koh}{Fournier-Viger et~al\mbox{.}}{2017}]%
        {fournier2017survey}
\bibfield{author}{\bibinfo{person}{Philippe Fournier-Viger},
  \bibinfo{person}{Jerry Chun-Wei Lin}, \bibinfo{person}{Rage-Uday Kiran},
  {and} \bibinfo{person}{Yun-Sing Koh}.} \bibinfo{year}{2017}\natexlab{}.
\newblock \showarticletitle{A survey of sequential pattern mining}.
\newblock \bibinfo{journal}{\emph{Data Science and Pattern Recognition}}
  \bibinfo{volume}{1}, \bibinfo{number}{1} (\bibinfo{year}{2017}),
  \bibinfo{pages}{54--77}.
\newblock


\bibitem[\protect\citeauthoryear{Fournier-Viger, Wu, Zida, and
  Tseng}{Fournier-Viger et~al\mbox{.}}{2014}]%
        {fournier2014fhm}
\bibfield{author}{\bibinfo{person}{Philippe Fournier-Viger},
  \bibinfo{person}{Cheng-Wei Wu}, \bibinfo{person}{Souleymane Zida}, {and}
  \bibinfo{person}{Vincent~S Tseng}.} \bibinfo{year}{2014}\natexlab{}.
\newblock \showarticletitle{{FHM}: Faster high-utility itemset mining using
  estimated utility co-occurrence pruning}. In
  \bibinfo{booktitle}{\emph{International Symposium on Methodologies for
  Intelligent Systems}}. Springer, \bibinfo{pages}{83--92}.
\newblock


\bibitem[\protect\citeauthoryear{Gan, Lin, Chao, Wang, and Yu}{Gan
  et~al\mbox{.}}{2018a}]%
        {gan2018privacy}
\bibfield{author}{\bibinfo{person}{Wensheng Gan}, \bibinfo{person}{Jerry
  Chun-Wei Lin}, \bibinfo{person}{Han-Chieh Chao}, \bibinfo{person}{Shyue-Liang
  Wang}, {and} \bibinfo{person}{Philip~S Yu}.}
  \bibinfo{year}{2018}\natexlab{a}.
\newblock \showarticletitle{Privacy preserving utility mining: a survey}. In
  \bibinfo{booktitle}{\emph{Proceedings of the IEEE International Conference on
  Big Data}}. IEEE, \bibinfo{pages}{2617--2626}.
\newblock


\bibitem[\protect\citeauthoryear{Gan, Lin, Chao, and Zhan}{Gan
  et~al\mbox{.}}{2017}]%
        {4gan2017data}
\bibfield{author}{\bibinfo{person}{Wensheng Gan}, \bibinfo{person}{Jerry
  Chun-Wei Lin}, \bibinfo{person}{Han-Chieh Chao}, {and}
  \bibinfo{person}{Justin Zhan}.} \bibinfo{year}{2017}\natexlab{}.
\newblock \showarticletitle{Data mining in distributed environment: a survey}.
\newblock \bibinfo{journal}{\emph{Wiley Interdisciplinary Reviews: Data Mining
  and Knowledge Discovery}} \bibinfo{volume}{7}, \bibinfo{number}{6}
  (\bibinfo{year}{2017}), \bibinfo{pages}{e1216}.
\newblock


\bibitem[\protect\citeauthoryear{Gan, Lin, Fournier-Viger, and Chao}{Gan
  et~al\mbox{.}}{2016}]%
        {gan2016more}
\bibfield{author}{\bibinfo{person}{Wensheng Gan}, \bibinfo{person}{Jerry
  Chun-Wei Lin}, \bibinfo{person}{Philippe Fournier-Viger}, {and}
  \bibinfo{person}{Han-Chieh Chao}.} \bibinfo{year}{2016}\natexlab{}.
\newblock \showarticletitle{More efficient algorithms for mining high-utility
  itemsets with multiple minimum utility thresholds}. In
  \bibinfo{booktitle}{\emph{International Conference on Database and Expert
  Systems Applications}}. Springer, \bibinfo{pages}{71--87}.
\newblock


\bibitem[\protect\citeauthoryear{Gan, Lin, Fournier-Viger, Chao, Hong, and
  Fujita}{Gan et~al\mbox{.}}{2018b}]%
        {1gan2018survey}
\bibfield{author}{\bibinfo{person}{Wensheng Gan}, \bibinfo{person}{Jerry
  Chun-Wei Lin}, \bibinfo{person}{Philippe Fournier-Viger},
  \bibinfo{person}{Han-Chieh Chao}, \bibinfo{person}{Tzung-Pei Hong}, {and}
  \bibinfo{person}{Hamido Fujita}.} \bibinfo{year}{2018}\natexlab{b}.
\newblock \showarticletitle{A survey of incremental high-utility itemset
  mining}.
\newblock \bibinfo{journal}{\emph{Wiley Interdisciplinary Reviews: Data Mining
  and Knowledge Discovery}} \bibinfo{volume}{8}, \bibinfo{number}{2}
  (\bibinfo{year}{2018}), \bibinfo{pages}{e1242}.
\newblock


\bibitem[\protect\citeauthoryear{Gan, Lin, Fournier-Viger, Chao, Tseng, and
  Yu}{Gan et~al\mbox{.}}{2019c}]%
        {2gan2018survey}
\bibfield{author}{\bibinfo{person}{Wensheng Gan}, \bibinfo{person}{Jerry
  Chun-Wei Lin}, \bibinfo{person}{Philippe Fournier-Viger},
  \bibinfo{person}{Han-Chieh Chao}, \bibinfo{person}{Vincent~S Tseng}, {and}
  \bibinfo{person}{Philip~S Yu}.} \bibinfo{year}{2019}\natexlab{c}.
\newblock \showarticletitle{A survey of utility-oriented pattern mining}.
\newblock \bibinfo{journal}{\emph{IEEE Transactions on Knowledge and Data
  Engineering, DOI: 10.1109/TKDE.2019.2942594}} (\bibinfo{year}{2019}),
  \bibinfo{pages}{1--22}.
\newblock


\bibitem[\protect\citeauthoryear{Gan, Lin, Fournier-Viger, Chao, and Yu}{Gan
  et~al\mbox{.}}{2019a}]%
        {gan2019huopm}
\bibfield{author}{\bibinfo{person}{Wensheng Gan}, \bibinfo{person}{Jerry
  Chun-Wei Lin}, \bibinfo{person}{Philippe Fournier-Viger},
  \bibinfo{person}{Han-Chieh Chao}, {and} \bibinfo{person}{Philip~S Yu}.}
  \bibinfo{year}{2019}\natexlab{a}.
\newblock \showarticletitle{{HUOPM}: High-utility occupancy pattern mining}.
\newblock \bibinfo{journal}{\emph{IEEE Transactions on Cybernetics, DOI:
  10.1109/TCYB.2019.2896267}} (\bibinfo{year}{2019}), \bibinfo{pages}{1--14}.
\newblock


\bibitem[\protect\citeauthoryear{Gan, Lin, Fournier-Viger, Chao, and Yu}{Gan
  et~al\mbox{.}}{2019b}]%
        {3gan2018survey}
\bibfield{author}{\bibinfo{person}{Wensheng Gan}, \bibinfo{person}{Jerry
  Chun-Wei Lin}, \bibinfo{person}{Philippe Fournier-Viger},
  \bibinfo{person}{Han-Chieh Chao}, {and} \bibinfo{person}{Philip~S Yu}.}
  \bibinfo{year}{2019}\natexlab{b}.
\newblock \showarticletitle{A survey of parallel sequential pattern mining}.
\newblock \bibinfo{journal}{\emph{ACM Transactions on Knowledge Discovery from
  Data}} \bibinfo{volume}{13}, \bibinfo{number}{3} (\bibinfo{year}{2019}),
  \bibinfo{pages}{25}.
\newblock


\bibitem[\protect\citeauthoryear{Gan, Lin, Jiexiong, Chao, Fujita, and Yu}{Gan
  et~al\mbox{.}}{2019d}]%
        {gan2019proum}
\bibfield{author}{\bibinfo{person}{Wensheng Gan}, \bibinfo{person}{Jerry
  Chun-Wei Lin}, \bibinfo{person}{Zhang Jiexiong}, \bibinfo{person}{Han-Chieh
  Chao}, \bibinfo{person}{Hamido Fujita}, {and} \bibinfo{person}{Philip~S Yu}.}
  \bibinfo{year}{2019}\natexlab{d}.
\newblock \showarticletitle{Pro{UM}: High utility sequential pattern mining}.
  In \bibinfo{booktitle}{\emph{Proceedings of the IEEE International Conference
  on Systems, Man, and Cybernetics}}. IEEE, \bibinfo{pages}{767--773}.
\newblock


\bibitem[\protect\citeauthoryear{Han, Pei, Yin, and Mao}{Han
  et~al\mbox{.}}{2004}]%
        {han2004mining}
\bibfield{author}{\bibinfo{person}{Jiawei Han}, \bibinfo{person}{Jian Pei},
  \bibinfo{person}{Yiwen Yin}, {and} \bibinfo{person}{Runying Mao}.}
  \bibinfo{year}{2004}\natexlab{}.
\newblock \showarticletitle{Mining frequent patterns without candidate
  generation: A frequent-pattern tree approach}.
\newblock \bibinfo{journal}{\emph{Data Mining and Knowledge Discovery}}
  \bibinfo{volume}{8}, \bibinfo{number}{1} (\bibinfo{year}{2004}),
  \bibinfo{pages}{53--87}.
\newblock


\bibitem[\protect\citeauthoryear{Hu, Wu, and Liao}{Hu et~al\mbox{.}}{2013}]%
        {hu2013efficient}
\bibfield{author}{\bibinfo{person}{Ya-Han Hu}, \bibinfo{person}{Fan Wu}, {and}
  \bibinfo{person}{Yi-Jiun Liao}.} \bibinfo{year}{2013}\natexlab{}.
\newblock \showarticletitle{An efficient tree-based algorithm for mining
  sequential patterns with multiple minimum supports}.
\newblock \bibinfo{journal}{\emph{Journal of Systems and Software}}
  \bibinfo{volume}{86}, \bibinfo{number}{5} (\bibinfo{year}{2013}),
  \bibinfo{pages}{1224--1238}.
\newblock


\bibitem[\protect\citeauthoryear{Koh and Ravana}{Koh and Ravana}{2016}]%
        {koh2016unsupervised}
\bibfield{author}{\bibinfo{person}{Yun-Sing Koh} {and}
  \bibinfo{person}{Sri-Devi Ravana}.} \bibinfo{year}{2016}\natexlab{}.
\newblock \showarticletitle{Unsupervised rare pattern mining: a survey}.
\newblock \bibinfo{journal}{\emph{ACM Transactions on Knowledge Discovery from
  Data}} \bibinfo{volume}{10}, \bibinfo{number}{4} (\bibinfo{year}{2016}),
  \bibinfo{pages}{45}.
\newblock


\bibitem[\protect\citeauthoryear{Krishnamoorthy}{Krishnamoorthy}{2018}]%
        {krishnamoorthy2018efficient}
\bibfield{author}{\bibinfo{person}{Srikumar Krishnamoorthy}.}
  \bibinfo{year}{2018}\natexlab{}.
\newblock \showarticletitle{Efficient mining of high utility itemsets with
  multiple minimum utility thresholds}.
\newblock \bibinfo{journal}{\emph{Engineering Applications of Artificial
  Intelligence}}  \bibinfo{volume}{69} (\bibinfo{year}{2018}),
  \bibinfo{pages}{112--126}.
\newblock


\bibitem[\protect\citeauthoryear{Lan, Hong, and Tseng}{Lan
  et~al\mbox{.}}{2011}]%
        {lan2011discovery}
\bibfield{author}{\bibinfo{person}{Guo-Cheng Lan}, \bibinfo{person}{Tzung-Pei
  Hong}, {and} \bibinfo{person}{Vincent~S Tseng}.}
  \bibinfo{year}{2011}\natexlab{}.
\newblock \showarticletitle{Discovery of high utility itemsets from on-shelf
  time periods of products}.
\newblock \bibinfo{journal}{\emph{Expert Systems with Applications}}
  \bibinfo{volume}{38}, \bibinfo{number}{5} (\bibinfo{year}{2011}),
  \bibinfo{pages}{5851--5857}.
\newblock


\bibitem[\protect\citeauthoryear{Lan, Hong, Tseng, and Wang}{Lan
  et~al\mbox{.}}{2014}]%
        {lan2014applying}
\bibfield{author}{\bibinfo{person}{Guo-Cheng Lan}, \bibinfo{person}{Tzung-Pei
  Hong}, \bibinfo{person}{Vincent~S Tseng}, {and} \bibinfo{person}{Shyue-Liang
  Wang}.} \bibinfo{year}{2014}\natexlab{}.
\newblock \showarticletitle{Applying the maximum utility measure in high
  utility sequential pattern mining}.
\newblock \bibinfo{journal}{\emph{Expert Systems with Applications}}
  \bibinfo{volume}{41}, \bibinfo{number}{11} (\bibinfo{year}{2014}),
  \bibinfo{pages}{5071--5081}.
\newblock


\bibitem[\protect\citeauthoryear{Lee, Wu, Lee, Liu, and Chen}{Lee
  et~al\mbox{.}}{2009}]%
        {lee2009mining}
\bibfield{author}{\bibinfo{person}{Anthony~JT Lee}, \bibinfo{person}{Huei-Wen
  Wu}, \bibinfo{person}{Tzu-Yu Lee}, \bibinfo{person}{Ying-Ho Liu}, {and}
  \bibinfo{person}{Kuo-Tay Chen}.} \bibinfo{year}{2009}\natexlab{}.
\newblock \showarticletitle{Mining closed patterns in multi-sequence
  time-series databases}.
\newblock \bibinfo{journal}{\emph{Data \& Knowledge Engineering}}
  \bibinfo{volume}{68}, \bibinfo{number}{10} (\bibinfo{year}{2009}),
  \bibinfo{pages}{1071--1090}.
\newblock


\bibitem[\protect\citeauthoryear{Lin, Hong, and Lu}{Lin et~al\mbox{.}}{2011}]%
        {lin2011effective}
\bibfield{author}{\bibinfo{person}{Chun-Wei Lin}, \bibinfo{person}{Tzung-Pei
  Hong}, {and} \bibinfo{person}{Wen-Hsiang Lu}.}
  \bibinfo{year}{2011}\natexlab{}.
\newblock \showarticletitle{An effective tree structure for mining high utility
  itemsets}.
\newblock \bibinfo{journal}{\emph{Expert Systems with Applications}}
  \bibinfo{volume}{38}, \bibinfo{number}{6} (\bibinfo{year}{2011}),
  \bibinfo{pages}{7419--7424}.
\newblock


\bibitem[\protect\citeauthoryear{Lin, Gan, Fournier-Viger, Hong, and Tseng}{Lin
  et~al\mbox{.}}{2016a}]%
        {lin2016efficient}
\bibfield{author}{\bibinfo{person}{Jerry Chun-Wei Lin},
  \bibinfo{person}{Wensheng Gan}, \bibinfo{person}{Philippe Fournier-Viger},
  \bibinfo{person}{Tzung-Pei Hong}, {and} \bibinfo{person}{Vincent~S Tseng}.}
  \bibinfo{year}{2016}\natexlab{a}.
\newblock \showarticletitle{Efficient algorithms for mining high-utility
  itemsets in uncertain databases}.
\newblock \bibinfo{journal}{\emph{Knowledge-Based Systems}}
  \bibinfo{volume}{96} (\bibinfo{year}{2016}), \bibinfo{pages}{171--187}.
\newblock


\bibitem[\protect\citeauthoryear{Lin, Gan, Fournier-Viger, Hong, and Zhan}{Lin
  et~al\mbox{.}}{2016b}]%
        {2lin2016efficient}
\bibfield{author}{\bibinfo{person}{Jerry Chun-Wei Lin},
  \bibinfo{person}{Wensheng Gan}, \bibinfo{person}{Philippe Fournier-Viger},
  \bibinfo{person}{Tzung-Pei Hong}, {and} \bibinfo{person}{Justin Zhan}.}
  \bibinfo{year}{2016}\natexlab{b}.
\newblock \showarticletitle{Efficient mining of high-utility itemsets using
  multiple minimum utility thresholds}.
\newblock \bibinfo{journal}{\emph{Knowledge-Based Systems}}
  \bibinfo{volume}{113} (\bibinfo{year}{2016}), \bibinfo{pages}{100--115}.
\newblock


\bibitem[\protect\citeauthoryear{Lin, Gan, and Hong}{Lin
  et~al\mbox{.}}{2015a}]%
        {lin2015fast}
\bibfield{author}{\bibinfo{person}{Jerry Chun-Wei Lin},
  \bibinfo{person}{Wensheng Gan}, {and} \bibinfo{person}{Tzung-Pei Hong}.}
  \bibinfo{year}{2015}\natexlab{a}.
\newblock \showarticletitle{A fast updated algorithm to maintain the discovered
  high-utility itemsets for transaction modification}.
\newblock \bibinfo{journal}{\emph{Advanced Engineering Informatics}}
  \bibinfo{volume}{29}, \bibinfo{number}{3} (\bibinfo{year}{2015}),
  \bibinfo{pages}{562--574}.
\newblock


\bibitem[\protect\citeauthoryear{Lin, Gan, Hong, and Tseng}{Lin
  et~al\mbox{.}}{2015b}]%
        {lin2015efficient}
\bibfield{author}{\bibinfo{person}{Jerry Chun-Wei Lin},
  \bibinfo{person}{Wensheng Gan}, \bibinfo{person}{Tzung-Pei Hong}, {and}
  \bibinfo{person}{Vincent~S Tseng}.} \bibinfo{year}{2015}\natexlab{b}.
\newblock \showarticletitle{Efficient algorithms for mining up-to-date
  high-utility patterns}.
\newblock \bibinfo{journal}{\emph{Advanced Engineering Informatics}}
  \bibinfo{volume}{29}, \bibinfo{number}{3} (\bibinfo{year}{2015}),
  \bibinfo{pages}{648--661}.
\newblock


\bibitem[\protect\citeauthoryear{Lin, Zhang, and Fournier-Viger}{Lin
  et~al\mbox{.}}{2017}]%
        {lin2017high}
\bibfield{author}{\bibinfo{person}{Jerry Chun-Wei Lin},
  \bibinfo{person}{Jiexiong Zhang}, {and} \bibinfo{person}{Philippe
  Fournier-Viger}.} \bibinfo{year}{2017}\natexlab{}.
\newblock \showarticletitle{High-utility sequential pattern mining with
  multiple minimum utility thresholds}. In
  \bibinfo{booktitle}{\emph{Asia-Pacific Web and Web-Age Information Management
  Joint Conference on Web and Big Data}}. Springer, \bibinfo{pages}{215--229}.
\newblock


\bibitem[\protect\citeauthoryear{Liu, Hsu, and Ma}{Liu et~al\mbox{.}}{1999}]%
        {liu1999mining}
\bibfield{author}{\bibinfo{person}{Bing Liu}, \bibinfo{person}{Wynne Hsu},
  {and} \bibinfo{person}{Yiming Ma}.} \bibinfo{year}{1999}\natexlab{}.
\newblock \showarticletitle{Mining association rules with multiple minimum
  supports}. In \bibinfo{booktitle}{\emph{Proceedings of the Fifth ACM SIGKDD
  International Conference on Knowledge Discovery and Data Mining}}. ACM,
  \bibinfo{pages}{337--341}.
\newblock


\bibitem[\protect\citeauthoryear{Liu and Qu}{Liu and Qu}{2012}]%
        {liu2012mining}
\bibfield{author}{\bibinfo{person}{Mengchi Liu} {and} \bibinfo{person}{Junfeng
  Qu}.} \bibinfo{year}{2012}\natexlab{}.
\newblock \showarticletitle{Mining high utility itemsets without candidate
  generation}. In \bibinfo{booktitle}{\emph{Proceedings of the 21st ACM
  International Conference on Information and Knowledge Management}}. ACM,
  \bibinfo{pages}{55--64}.
\newblock


\bibitem[\protect\citeauthoryear{Liu, Liao, and Choudhary}{Liu
  et~al\mbox{.}}{2005}]%
        {liu2005two}
\bibfield{author}{\bibinfo{person}{Ying Liu}, \bibinfo{person}{Wei-Keng Liao},
  {and} \bibinfo{person}{Alok Choudhary}.} \bibinfo{year}{2005}\natexlab{}.
\newblock \showarticletitle{A two-phase algorithm for fast discovery of high
  utility itemsets}. In \bibinfo{booktitle}{\emph{Pacific-Asia Conference on
  Knowledge Discovery and Data Mining}}. Springer, \bibinfo{pages}{689--695}.
\newblock


\bibitem[\protect\citeauthoryear{Liu, Cheng, and Tseng}{Liu
  et~al\mbox{.}}{2011}]%
        {liu2011discovering}
\bibfield{author}{\bibinfo{person}{Yu-Cheng Liu}, \bibinfo{person}{Chun-Pei
  Cheng}, {and} \bibinfo{person}{Vincent~S Tseng}.}
  \bibinfo{year}{2011}\natexlab{}.
\newblock \showarticletitle{Discovering relational-based association rules with
  multiple minimum supports on microarray datasets}.
\newblock \bibinfo{journal}{\emph{Bioinformatics}} \bibinfo{volume}{27},
  \bibinfo{number}{22} (\bibinfo{year}{2011}), \bibinfo{pages}{3142--3148}.
\newblock


\bibitem[\protect\citeauthoryear{Mai, Vo, and Nguyen}{Mai
  et~al\mbox{.}}{2017}]%
        {mai2017lattice}
\bibfield{author}{\bibinfo{person}{Thang Mai}, \bibinfo{person}{Bay Vo}, {and}
  \bibinfo{person}{Loan~TT Nguyen}.} \bibinfo{year}{2017}\natexlab{}.
\newblock \showarticletitle{A lattice-based approach for mining high utility
  association rules}.
\newblock \bibinfo{journal}{\emph{Information Sciences}}  \bibinfo{volume}{399}
  (\bibinfo{year}{2017}), \bibinfo{pages}{81--97}.
\newblock


\bibitem[\protect\citeauthoryear{Mannila, Toivonen, and Verkamo}{Mannila
  et~al\mbox{.}}{1997}]%
        {mannila1997discovery}
\bibfield{author}{\bibinfo{person}{Heikki Mannila}, \bibinfo{person}{Hannu
  Toivonen}, {and} \bibinfo{person}{A~Inkeri Verkamo}.}
  \bibinfo{year}{1997}\natexlab{}.
\newblock \showarticletitle{Discovery of frequent episodes in event sequences}.
\newblock \bibinfo{journal}{\emph{Data Mining and Knowledge Discovery}}
  \bibinfo{volume}{1}, \bibinfo{number}{3} (\bibinfo{year}{1997}),
  \bibinfo{pages}{259--289}.
\newblock


\bibitem[\protect\citeauthoryear{Marshall}{Marshall}{2005}]%
        {marshall2005principles}
\bibfield{author}{\bibinfo{person}{Alfred Marshall}.}
  \bibinfo{year}{2005}\natexlab{}.
\newblock \showarticletitle{From Principles of Economics}.
\newblock In \bibinfo{booktitle}{\emph{Readings in the Economics of the
  Division of Labor: the Classical Tradition}}. \bibinfo{publisher}{World
  Scientific}, \bibinfo{pages}{195--215}.
\newblock


\bibitem[\protect\citeauthoryear{Pei, Han, Mortazavi-Asl, Pinto, Chen, Dayal,
  and Hsu}{Pei et~al\mbox{.}}{2001}]%
        {pei2001prefixspan}
\bibfield{author}{\bibinfo{person}{Jian Pei}, \bibinfo{person}{Jiawei Han},
  \bibinfo{person}{Behzad Mortazavi-Asl}, \bibinfo{person}{Helen Pinto},
  \bibinfo{person}{Qiming Chen}, \bibinfo{person}{Umeshwar Dayal}, {and}
  \bibinfo{person}{Mei~Chun Hsu}.} \bibinfo{year}{2001}\natexlab{}.
\newblock \showarticletitle{Prefix{S}pan: Mining sequential patterns
  efficiently by prefix-projected pattern growth}. In
  \bibinfo{booktitle}{\emph{The International Conference on Data Engineering}}.
  IEEE, \bibinfo{pages}{215--224}.
\newblock


\bibitem[\protect\citeauthoryear{Pinto, Han, Pei, Wang, Chen, and Dayal}{Pinto
  et~al\mbox{.}}{2001}]%
        {pinto2001multi}
\bibfield{author}{\bibinfo{person}{Helen Pinto}, \bibinfo{person}{Jiawei Han},
  \bibinfo{person}{Jian Pei}, \bibinfo{person}{Ke Wang},
  \bibinfo{person}{Qiming Chen}, {and} \bibinfo{person}{Umeshwar Dayal}.}
  \bibinfo{year}{2001}\natexlab{}.
\newblock \showarticletitle{Multi-dimensional sequential pattern mining}. In
  \bibinfo{booktitle}{\emph{Proceedings of the 10th International Conference on
  Information and Knowledge Management}}. ACM, \bibinfo{pages}{81--88}.
\newblock


\bibitem[\protect\citeauthoryear{Ryang and Yun}{Ryang and Yun}{2016}]%
        {ryang2016high}
\bibfield{author}{\bibinfo{person}{Heungmo Ryang} {and} \bibinfo{person}{Unil
  Yun}.} \bibinfo{year}{2016}\natexlab{}.
\newblock \showarticletitle{High utility pattern mining over data streams with
  sliding window technique}.
\newblock \bibinfo{journal}{\emph{Expert Systems with Applications}}
  \bibinfo{volume}{57} (\bibinfo{year}{2016}), \bibinfo{pages}{214--231}.
\newblock


\bibitem[\protect\citeauthoryear{Shie, Hsiao, Tseng, and Yu}{Shie
  et~al\mbox{.}}{2011}]%
        {shie2011mining}
\bibfield{author}{\bibinfo{person}{Bai-En Shie}, \bibinfo{person}{Hui-Fang
  Hsiao}, \bibinfo{person}{Vincent~S Tseng}, {and} \bibinfo{person}{Philip~S
  Yu}.} \bibinfo{year}{2011}\natexlab{}.
\newblock \showarticletitle{Mining high utility mobile sequential patterns in
  mobile commerce environments}. In \bibinfo{booktitle}{\emph{Proceedings of
  International Conference on Database Systems for Advanced Applications}}.
  Springer, \bibinfo{pages}{224--238}.
\newblock


\bibitem[\protect\citeauthoryear{Srikant and Agrawal}{Srikant and
  Agrawal}{1996}]%
        {srikant1996mining}
\bibfield{author}{\bibinfo{person}{Ramakrishnan Srikant} {and}
  \bibinfo{person}{Rakesh Agrawal}.} \bibinfo{year}{1996}\natexlab{}.
\newblock \showarticletitle{Mining sequential patterns: generalizations and
  performance improvements}. In \bibinfo{booktitle}{\emph{Proceedings of
  International Conference on Extending Database Technology}}. Springer,
  \bibinfo{pages}{1--17}.
\newblock


\bibitem[\protect\citeauthoryear{Tseng, Shie, Wu, and Yu}{Tseng
  et~al\mbox{.}}{2013}]%
        {tseng2013efficient}
\bibfield{author}{\bibinfo{person}{Vincent~S Tseng}, \bibinfo{person}{Bai-En
  Shie}, \bibinfo{person}{Cheng-Wei Wu}, {and} \bibinfo{person}{Philip~S Yu}.}
  \bibinfo{year}{2013}\natexlab{}.
\newblock \showarticletitle{Efficient algorithms for mining high utility
  itemsets from transactional databases}.
\newblock \bibinfo{journal}{\emph{IEEE Transactions on Knowledge and Data
  Engineering}} \bibinfo{volume}{25}, \bibinfo{number}{8}
  (\bibinfo{year}{2013}), \bibinfo{pages}{1772--1786}.
\newblock


\bibitem[\protect\citeauthoryear{Tseng, Wu, Fournier-Viger, and Yu}{Tseng
  et~al\mbox{.}}{2015}]%
        {tseng2015efficient}
\bibfield{author}{\bibinfo{person}{Vincent~S Tseng}, \bibinfo{person}{Cheng-Wei
  Wu}, \bibinfo{person}{Philippe Fournier-Viger}, {and}
  \bibinfo{person}{Philip~S Yu}.} \bibinfo{year}{2015}\natexlab{}.
\newblock \showarticletitle{Efficient algorithms for mining the concise and
  lossless representation of high utility itemsets}.
\newblock \bibinfo{journal}{\emph{IEEE Transactions on Knowledge and Data
  Engineering}} \bibinfo{volume}{27}, \bibinfo{number}{3}
  (\bibinfo{year}{2015}), \bibinfo{pages}{726--739}.
\newblock


\bibitem[\protect\citeauthoryear{Tseng, Wu, Fournier-Viger, and Yu}{Tseng
  et~al\mbox{.}}{2016}]%
        {tseng2016efficient}
\bibfield{author}{\bibinfo{person}{Vincent~S Tseng}, \bibinfo{person}{Cheng-Wei
  Wu}, \bibinfo{person}{Philippe Fournier-Viger}, {and}
  \bibinfo{person}{Philip~S Yu}.} \bibinfo{year}{2016}\natexlab{}.
\newblock \showarticletitle{Efficient algorithms for mining top-$k$ high
  utility itemsets}.
\newblock \bibinfo{journal}{\emph{IEEE Transactions on Knowledge and Data
  Engineering}} \bibinfo{volume}{28}, \bibinfo{number}{1}
  (\bibinfo{year}{2016}), \bibinfo{pages}{54--67}.
\newblock


\bibitem[\protect\citeauthoryear{Tseng, Wu, Shie, and Yu}{Tseng
  et~al\mbox{.}}{2010}]%
        {tseng2010up}
\bibfield{author}{\bibinfo{person}{Vincent~S Tseng}, \bibinfo{person}{Cheng-Wei
  Wu}, \bibinfo{person}{Bai-En Shie}, {and} \bibinfo{person}{Philip~S Yu}.}
  \bibinfo{year}{2010}\natexlab{}.
\newblock \showarticletitle{{UP-G}rowth: an efficient algorithm for high
  utility itemset mining}. In \bibinfo{booktitle}{\emph{Proceedings of the 16th
  ACM SIGKDD International Conference on Knowledge Discovery and Data Mining}}.
  ACM, \bibinfo{pages}{253--262}.
\newblock


\bibitem[\protect\citeauthoryear{Wang and Huang}{Wang and Huang}{2018}]%
        {wang2018incremental}
\bibfield{author}{\bibinfo{person}{Jun-Zhe Wang} {and}
  \bibinfo{person}{Jiun-Long Huang}.} \bibinfo{year}{2018}\natexlab{}.
\newblock \showarticletitle{On incremental high utility sequential pattern
  mining}.
\newblock \bibinfo{journal}{\emph{ACM Transactions on Intelligent Systems and
  Technology}} \bibinfo{volume}{9}, \bibinfo{number}{5} (\bibinfo{year}{2018}),
  \bibinfo{pages}{55}.
\newblock


\bibitem[\protect\citeauthoryear{Wang, Huang, and Chen}{Wang
  et~al\mbox{.}}{2016}]%
        {wang2016efficiently}
\bibfield{author}{\bibinfo{person}{Jun-Zhe Wang}, \bibinfo{person}{Jiun-Long
  Huang}, {and} \bibinfo{person}{Yi-Cheng Chen}.}
  \bibinfo{year}{2016}\natexlab{}.
\newblock \showarticletitle{On efficiently mining high utility sequential
  patterns}.
\newblock \bibinfo{journal}{\emph{Knowledge and Information Systems}}
  \bibinfo{volume}{49}, \bibinfo{number}{2} (\bibinfo{year}{2016}),
  \bibinfo{pages}{597--627}.
\newblock


\bibitem[\protect\citeauthoryear{Yao, Hamilton, and Butz}{Yao
  et~al\mbox{.}}{2004}]%
        {yao2004foundational}
\bibfield{author}{\bibinfo{person}{Hong Yao}, \bibinfo{person}{Howard~J
  Hamilton}, {and} \bibinfo{person}{Cory~J Butz}.}
  \bibinfo{year}{2004}\natexlab{}.
\newblock \showarticletitle{A foundational approach to mining itemset utilities
  from databases}. In \bibinfo{booktitle}{\emph{Proceedings of the SIAM
  International Conference on Data Mining}}. SIAM, \bibinfo{pages}{482--486}.
\newblock


\bibitem[\protect\citeauthoryear{Yin, Cui, Sun, Hu, and Chen}{Yin
  et~al\mbox{.}}{2014}]%
        {yin2014lcars}
\bibfield{author}{\bibinfo{person}{Hongzhi Yin}, \bibinfo{person}{Bin Cui},
  \bibinfo{person}{Yizhou Sun}, \bibinfo{person}{Zhiting Hu}, {and}
  \bibinfo{person}{Ling Chen}.} \bibinfo{year}{2014}\natexlab{}.
\newblock \showarticletitle{LCARS: A spatial item recommender system}.
\newblock \bibinfo{journal}{\emph{ACM Transactions on Information Systems}}
  \bibinfo{volume}{32}, \bibinfo{number}{3} (\bibinfo{year}{2014}),
  \bibinfo{pages}{11}.
\newblock


\bibitem[\protect\citeauthoryear{Yin, Zheng, and Cao}{Yin
  et~al\mbox{.}}{2012}]%
        {yin2012uspan}
\bibfield{author}{\bibinfo{person}{Junfu Yin}, \bibinfo{person}{Zhigang Zheng},
  {and} \bibinfo{person}{Longbing Cao}.} \bibinfo{year}{2012}\natexlab{}.
\newblock \showarticletitle{{US}pan: an efficient algorithm for mining high
  utility sequential patterns}. In \bibinfo{booktitle}{\emph{Proceedings of the
  18th ACM SIGKDD International Conference on Knowledge Discovery and Data
  Mining}}. ACM, \bibinfo{pages}{660--668}.
\newblock


\bibitem[\protect\citeauthoryear{Yin, Zheng, Cao, Song, and Wei}{Yin
  et~al\mbox{.}}{2013}]%
        {yin2013efficiently}
\bibfield{author}{\bibinfo{person}{Junfu Yin}, \bibinfo{person}{Zhigang Zheng},
  \bibinfo{person}{Longbing Cao}, \bibinfo{person}{Yin Song}, {and}
  \bibinfo{person}{Wei Wei}.} \bibinfo{year}{2013}\natexlab{}.
\newblock \showarticletitle{Efficiently mining top-$k$ high utility sequential
  patterns}. In \bibinfo{booktitle}{\emph{Proceedings of the IEEE 13th
  International Conference on Data Mining}}. IEEE, \bibinfo{pages}{1259--1264}.
\newblock


\bibitem[\protect\citeauthoryear{Yu and Chen}{Yu and Chen}{2005}]%
        {yu2005mining}
\bibfield{author}{\bibinfo{person}{Chung-Ching Yu} {and}
  \bibinfo{person}{Yen-Liang Chen}.} \bibinfo{year}{2005}\natexlab{}.
\newblock \showarticletitle{Mining sequential patterns from multidimensional
  sequence data}.
\newblock \bibinfo{journal}{\emph{IEEE Transactions on Knowledge and Data
  Engineering}} \bibinfo{volume}{17}, \bibinfo{number}{1}
  (\bibinfo{year}{2005}), \bibinfo{pages}{136--140}.
\newblock


\bibitem[\protect\citeauthoryear{Yun, Lee, and Yoon}{Yun et~al\mbox{.}}{2017}]%
        {yun2017efficient}
\bibfield{author}{\bibinfo{person}{Unil Yun}, \bibinfo{person}{Gangin Lee},
  {and} \bibinfo{person}{Eunchul Yoon}.} \bibinfo{year}{2017}\natexlab{}.
\newblock \showarticletitle{Efficient high utility pattern mining for
  establishing manufacturing plans with sliding window control}.
\newblock \bibinfo{journal}{\emph{IEEE Transactions on Industrial Electronics}}
  \bibinfo{volume}{64}, \bibinfo{number}{9} (\bibinfo{year}{2017}),
  \bibinfo{pages}{7239--7249}.
\newblock


\bibitem[\protect\citeauthoryear{Zaki}{Zaki}{2001}]%
        {zaki2001spade}
\bibfield{author}{\bibinfo{person}{Mohammed~J Zaki}.}
  \bibinfo{year}{2001}\natexlab{}.
\newblock \showarticletitle{{SPADE}: an efficient algorithm for mining frequent
  sequences}.
\newblock \bibinfo{journal}{\emph{Machine Learning}} \bibinfo{volume}{42},
  \bibinfo{number}{1-2} (\bibinfo{year}{2001}), \bibinfo{pages}{31--60}.
\newblock


\bibitem[\protect\citeauthoryear{Zida, Fournier-Viger, Lin, Wu, and Tseng}{Zida
  et~al\mbox{.}}{2015}]%
        {zida2015efim}
\bibfield{author}{\bibinfo{person}{Souleymane Zida}, \bibinfo{person}{Philippe
  Fournier-Viger}, \bibinfo{person}{Jerry Chun-Wei Lin},
  \bibinfo{person}{Cheng-Wei Wu}, {and} \bibinfo{person}{Vincent~S Tseng}.}
  \bibinfo{year}{2015}\natexlab{}.
\newblock \showarticletitle{{EFIM}: a highly efficient algorithm for
  high-utility itemset mining}. In \bibinfo{booktitle}{\emph{Mexican
  International Conference on Artificial Intelligence}}. Springer,
  \bibinfo{pages}{530--546}.
\newblock


\end{thebibliography}

\end{document}